\documentclass{article} % For LaTeX2e
\usepackage{myStyle_4_arXiv,times}
\usepackage{amsmath,amsfonts,bm}

\def\eqref#1{equation~\ref{#1}}
\def\1{\bm{1}}

\DeclareMathAlphabet{\mathsfit}{\encodingdefault}{\sfdefault}{m}{sl}
\SetMathAlphabet{\mathsfit}{bold}{\encodingdefault}{\sfdefault}{bx}{n}

\usepackage{hyperref,url}
\usepackage{amsmath,amssymb,dsfont,amsthm}
\usepackage{booktabs}
\usepackage{verbatim}
\usepackage{graphicx,wrapfig}
\usepackage{appendix}
\usepackage{subeqnarray}
\usepackage{chngcntr}
\usepackage{color}
\usepackage{enumitem}
\usepackage{caption}
\usepackage{subcaption}
\newtheorem{proposition}{\bf Proposition}[section]
\newtheorem{lemma}{\bf Lemma}[section]

\newcommand{\bdEDIT}[1]{\textcolor{black}{#1}}
\newcommand{\kgEDIT}[1]{\textcolor{black}{#1}}
\newcommand{\RR}{\mathds{R}}

\newcommand{\rUnit}{\frac{\bf r}{\left|{\bf r}\right|}}
\newcommand{\bUnit}{\frac{\hat{\bf b}}{|\hat{\bf b}|}}

\newcommand{\riBUnit}{\frac{{\bf r}_{ibi}}{\left|{\bf r}_{ibi}\right|}}
\newcommand{\rOneBUnit}{\frac{{\bf r}_{1b1}}{\left|{\bf r}_{1b1}\right|}}
\newcommand{\rOneBTwoUnit}{\frac{{\bf r}_{1b2}}{\left|{\bf r}_{1b2}\right|}}
\newcommand{\rOneBiUnit}{\frac{{\bf r}_{1bi}}{\left|{\bf r}_{1bi}\right|}}
\newcommand{\rTwoBUnit}{\frac{{\bf r}_{2b2}}{\left|{\bf r}_{2b2}\right|}}

\newcommand{\rijUnit}{\frac{{\bf r}_{ij}}{\left|{\bf r}_{ij}\right|}}
\DeclareMathOperator{\sgn}{sgn}

\title{Beacon-referenced Pursuit for Collective Motions in Three Dimensions}
\author{Kevin S.~Galloway
\\
Department of Electrical and Computer Engineering\\
United States Naval Academy\\
Annapolis, MD 21402, USA.\\
\texttt{kgallowa@usna.edu}
\And
Biswadip Dey\thanks{This work was conceptualized during Biswadip's previous appointment at the Department of Mechanical and Aerospace Engineering at Princeton University, Princeton, NJ 08540, USA.}\\
Predictive Analytics Research Group\\
Siemens Corporation, Corporate Technology \\
Princeton, NJ 08540, USA \\
\texttt{biswadip.dey@siemens.com}
}

\iclrfinalcopy 
\begin{document}

\maketitle

\begin{abstract}
Motivated by real-world applications of unmanned aerial vehicles, this paper introduces a decentralized control mechanism to guide steering control of autonomous agents maneuvering in the vicinity of multiple moving entities (e.g. other autonomous agents) and stationary entities (e.g. fixed beacons or points of references) in a three-dimensional environment. The proposed control law, which can be perceived as a modification of the three-dimensional constant bearing (CB) pursuit law, provides a means to allocate simultaneous attention to multiple entities. We investigate the behavior of the closed-loop dynamics for a system with one agent referencing two beacons, as well as a two-agent \emph{mutual pursuit} system wherein each agent employs the beacon-referenced CB pursuit law with regards to the other agent and a stationary beacon. Under certain assumptions on the associated control parameters, we demonstrate that this problem admits \emph{circling equilibria} with agents moving on circular orbits with a common radius, in planes perpendicular to a common axis passing through the beacons. As the common radius and distances from the beacon are determined by the choice of parameters in the pursuit law, this approach provides a means to engineer desired formations in a 3-dimensional setting.
\end{abstract}
\section{Introduction}
\label{sec:Intro}
As pursuit and collective motion play significant roles in various contexts of robotics and engineering, it seems appealing to seek inspiration from nature, which abounds with many such examples. Among the various possible ways to pursue and intercept a moving target, evidence of constant bearing (CB) pursuit strategy can be observed in a variety of animal species, such as flies \citep{Collett1975,Osorio281}, dogs \citep{Shaffer_Dog_2004}, raptors \citep{Tucker3745}, and humans \citep{CHARDENON200213}. While pursuing a target using CB pursuit strategy, the pursuer (i.e. the following agent) moves in such a way that the angle between its own velocity and the line-of-sight to the target (i.e. the baseline) remains constant. By prescribing a fixed offset between the baseline and the pursuer's velocity, this strategy provides a generalization of the classical pursuit strategy (wherein the pursuer's velocity is always aligned with the line-of-sight). While pursuit strategies are often employed in contexts with a single pursuer-target pair, recent work \citep{Marshall_TAC_04,Sinha_GenLCP,Pavone_ASME_commonOffset} has demonstrated that the CB pursuit strategy can be employed as a building block for designing coordinated maneuvers in a collective of agents by implementing the strategy in a cyclic manner (i.e. agent $i$ pursues agent $i+1$, with the last agent pursuing the first). In the planar setting, \cite{Kevin_2011_CDC, Galloway_PRS_13} demonstrated the existence and stability of a rich class of behaviors such as circular motion, rectilinear motion, shape preserving spirals and periodic orbits. In a related body of work, \cite{CP_3D_Kevin_2010,Galloway_PRS_16} defined a CB pursuit strategy and associated control law for the 3-dimensional setting and determined conditions for existence of a similar class of motions. While this line of research has demonstrated existence of circling equilibria in which agents moved on a common circular trajectory, both the location of the circumcenter of the formation (with respect to some inertial frame) and its size were determined by initial conditions rather than control parameters. To overcome this aspect and to broaden the scope from a design perspective, the work by \cite{KSG_BD_ACC_2015, KSG_BD_ACC_2016, KSG_BD_Automatica_18} introduced a modified version of the CB pursuit law (in the planar setting), wherein the pursuer pays attention to its neighbor as well as a stationary beacon. The beacon can represent an attractive food source in a biological setting, or some target of interest for an unmanned vehicle. (See also the work by \cite{Mallik2015ConsensusApplications} and \cite{Daingade2015AImplementation} for a related, but different, control formulation.) 

In the current work, we extend the beacon-referenced approach to a 3-dimensional setting and explore the possible equilibrium formations for two agents in beacon-referenced mutual pursuit. More specifically, we first state a modified version of the beacon-referenced 3D pursuit law introduced by \cite{KSG_BD_ACC_2018} which is based only on relative bearing measurements to two targets (typically a fixed beacon and a moving neighbor agent). This modified version of the control law is easier to implement since it does not require any estimate of optic flow or relative velocity, and it is more tractable for mathematical analysis. We then consider the corresponding dynamics for a mutual pursuit system in which two agents apply this pursuit law with respect to each other and to a (possibly distinct) fixed beacon. Earlier work \citep{Matteo_IFAC_3D, MISCHIATI2012894, UH_BD_ICRA_15} has demonstrated that mutual pursuit can lead to a variety of interesting motion patterns, while providing better tractability from a nonlinear analysis perspective, and it can be viewed as a building block towards understanding the more general cyclic pursuit framework. The main contribution of this work is to then derive existence conditions (and a mathematical characterization) for circling equilibria under the various possible combinations of 1-2 mobile agents interacting with 1-2 stationary beacons. 

This paper is organized as follows. We begin by describing the self-steering particle model for agents moving in three dimensions. Then, in the latter part of Section~\ref{sec:Prob_SetUp}, we introduce the 3D beacon-referenced constant bearing pursuit law. In the spirit of beginning with the configuration of least complexity, in Section \ref{sec:oneAgentTwoBeacons} we consider the case in which a single mobile agent employs the beacon-referenced control law with respect to two fixed beacons, and we characterize the existence of circling equilibria. Before proceeding to configurations involving two mobile agents, in Section \ref{sec:ShapeVarDescription} we describe a reduction to an \emph{effective shape space} which can be parametrized by corresponding scalar shape variables, and we state the closed-loop mutual pursuit dynamics in terms of these shape variables. Section \ref{sec:twoAgentsOneBeacon} addresses the case where two agents reference the same beacon, which is similar in spirit to the presentation by \cite{KSG_BD_ACC_2018} but incorporates the modified control law. Finally, in Section \ref{sec:twoAgentsTwoBeacons} we consider the most general case in which two agents employ the beacon-referenced pursuit law with regards to two distinct beacons. In the course of our analysis, we uncover the existence of a variety of coordinated 3D circling motions (as depicted in figures such as Figure~\ref{fig:TwoBeaconStackedCircle} and Figure \ref{fig:TwoBeaconCrossStackCircle}) for which pertinent attributes such as circling radius, angular separation of the agents, and vertical separation of the agents are all prescribed by control parameters rather than initial conditions. Since these control laws require only bearing measurements relative to the agent's forward velocity vector, we hypothesize that they could provide effective methods for station-keeping in a variety of contexts such as space exploration.

\section{Modeling Agent Dynamics and Control}
\label{sec:Prob_SetUp}
%%%%%%%%%%%%%%%%%%%%%%%%%%%%%%%%%%%%%%%%%%%%%%%%%%%%%%%%%%%%%%%%%%%%%%%%%%%%%%%%
\subsection{Generative Model: Agents as Self-Steering Particles}
%%%%%%%%%%%%%%%%%%%%%%%%%%%%%%%%%%%%%%%%%%%%%%%%%%%%%%%%%%%%%%%%%%%%%%%%%%%%%%%%
%
Similar to earlier works \citep{Justh_PSK_3Dformation, CP_3D_Kevin_2010}, we treat the agents as unit-mass self-steering particles moving along twice-differentiable paths in a 3-dimensional environment. This allows us to describe the motion of an agent in terms of its natural Frenet frame \citep{Nat_Frenet_Bishop}, defined by its position $\mathbf{r}_i \in \RR^3$ (with respect to an inertial reference frame) and an orthonormal triad of vectors $[\mathbf{x}_i,\mathbf{y}_i,\mathbf{z}_i]$. Then, by constraining the agents to move at a common fixed and nonvanishing speed, we can assume without loss of generality that the agents are moving at unit speed, and express the dynamics of a pair of agents as
\begin{equation}
\begin{aligned}
& 
\dot{\mathbf{r}}_i = \mathbf{x}_i \\
& 
\dot{\mathbf{x}}_i = u_i \mathbf{y}_i + v_i \mathbf{z}_i, \quad
\dot{\mathbf{y}}_i = - u_i \mathbf{x}_i, \quad
\dot{\mathbf{z}}_i = - v_i \mathbf{x}_i, 
\end{aligned} 
\label{Explicit_MODEL}
\end{equation}
for $i=1,2$. Here, $u_i$ and $v_i$ are the natural curvatures viewed as gyroscopic steering controls. We also define two (possibly colocated) fixed beacons ${\bf r}_{b1}$ and ${\bf r}_{b2}$, which can be assumed (without loss of generality) to be located at ${\bf r}_{b1}=(0,0,-b)$ and ${\bf r}_{b2}=(0,0,b)$ for some $b \geq 0$. The vector from beacon 1 to beacon 2 will be denoted as $\hat{\bf b} = \mathbf{r}_{b2}-\mathbf{r}_{b1} = (0,0,2b)$. We also denote the vectors $\mathbf{r} = \mathbf{r}_{1}-\mathbf{r}_{2}$,  $\mathbf{r}_{1b1} = \mathbf{r}_{1}-\mathbf{r}_{b1}$, and $\mathbf{r}_{2b2} = \mathbf{r}_{2}-\mathbf{r}_{b2}$ to express the relative configuration of the agents and the beacons, noting the resulting constraint
\begin{equation}
\label{eqn:closureConstraint}
    \hat{\bf b} +  \mathbf{r}_{2b2} + \mathbf{r} - \mathbf{r}_{1b1} = \mathbf{0}.
\end{equation}
These vector relationships are depicted in Figure \ref{fig:AgentVectorDiagram}. (Note that only the $\mathbf{x}_i$ vectors are depicted rather than the entire frame.)
%
%
%\begin{wrapfigure}[17]{R}{0.5\textwidth}
%\centering
%\includegraphics[height=2in]{TwoBeaconDiagramV3.png}
%\caption{Two mobile agents ${\bf r}_1$ and ${\bf r}_2$ maneuvering in the vicinity of fixed beacons ${\bf r}_{b1}$ and ${\bf r}_{b2}$.}
%\label{fig:AgentVectorDiagram}
%\vspace{-1.5em}
%\end{wrapfigure}
%
%
\begin{figure}[bh!]
\begin{center}
  \includegraphics[height=2in]{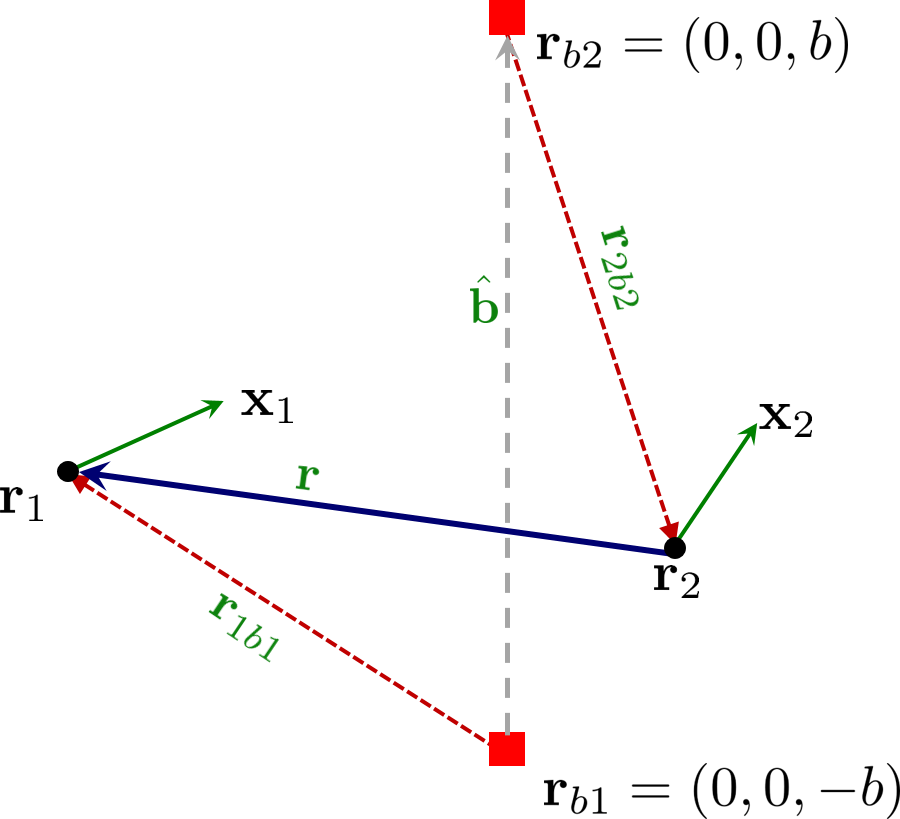}
  \caption{Two mobile agents ${\bf r}_1$ and ${\bf r}_2$ maneuvering in the vicinity of fixed beacons ${\bf r}_{b1}$ and ${\bf r}_{b2}$.}
  \vspace{-1.5em}
  \label{fig:AgentVectorDiagram}
\end{center}
\end{figure}
%
%
%
%%%%%%%%%%%%%%%%%%%%%%%%%%%%%%%%%%%%%%%%%%%%%%%%%%%%%%%%%%%%%%%%%%%%%%%%%%%%%%%%
\subsection{Beacon-referenced CB Pursuit in Three Dimensions}
%%%%%%%%%%%%%%%%%%%%%%%%%%%%%%%%%%%%%%%%%%%%%%%%%%%%%%%%%%%%%%%%%%%%%%%%%%%%%%%%
%
Previous work \citep{CP_3D_Kevin_2010,Galloway_PRS_16} introduced and analyzed a control law for executing the constant bearing (CB) pursuit strategy in three dimensions. Making use of the notation $\mathbf{r}_{ij} = \mathbf{r}_{i} - \mathbf{r}_{j}$ and letting $a_i \in [-1,1]$ be selected as the CB pursuit parameter for agent $i$, we say agent $i$ has attained the \textit{CB($a_i$) pursuit strategy} with respect to agent $j$ if ${\bf x}_{i} \cdot \rijUnit  = a_i$. The CB($a_i$) pursuit strategy can be thought of as prescribing a desired angle between the pursuer's velocity vector and the bearing vector to the desired target\footnote{We note that the bearing vector $\mathbf{r}_{ij}$ actually points \textit{towards} the pursuer, and therefore if $\alpha_i$ represents the desired CB pursuit angle, the corresponding CB($a_i$) parameter is given by $a_i = -\cos(\alpha_i)$.}. Then \emph{distance} from this pursuit state can be described by the CB cost function $\Lambda_{i} = \frac{1}{2}\biggl[\Bigl({\bf x}_{i} \cdot \rijUnit\Bigr) - a_i\biggr]^{2} \in [0,2]$, and a feedback law introduced by \cite{CP_3D_Kevin_2010,Galloway_PRS_16} was shown to enforce the property that $\dot{\Lambda}_i \leq 0$ with $\dot{\Lambda}_i = 0$ if and only if $\Lambda_i = 0$, (i.e. the control law renders the CB pursuit state attractive and positively invariant). The control law used by \cite{CP_3D_Kevin_2010,Galloway_PRS_16} is made up of one term referencing the bearing to the pursuee, and one term related to the relative velocity between the pursuer and pursuee. If we leave off the latter term, then the resulting simplified CB pursuit law, given by 
\begin{align}
\label{eqn:simplifiedCBLaw}
u_i = -\mu \left[\left({\bf x}_{i} \cdot \rijUnit\right) - a_i\right] \left({\bf y}_{i} \cdot \rijUnit\right)
; \quad
%\nonumber \\
v_i = -\mu  \left[\left({\bf x}_{i} \cdot \rijUnit\right) - a_i\right] \left({\bf z}_{i} \cdot \rijUnit\right)
\end{align}
lacks the invariance property but performs reasonably well in numerical simulations. (Here $\mu > 0$ is a control gain.) 

The simplified CB pursuit law given by \ref{eqn:simplifiedCBLaw} serves as our building block for developing a beacon-referenced CB pursuit law which uses only bearing measurements and which references two targets rather than only one. (For most of this work we will consider the case where one target is a maneuvering agent and the other target is a fixed beacon, but section \ref{sec:oneAgentTwoBeacons} addresses the case where a single mobile agent maneuvers with reference to two fixed beacons.)  Letting $\lambda \in [0,1]$ represent relative attention/priority between the targets, and $a_{bi} \in [-1,1]$ represent the ``CB-like" control parameter with reference to beacon $i$, we can express our beacon-referenced control law as
\begin{align}
\label{eqn:beaconReferencedCBLaw}
u_i &= -\mu (1-\lambda) \left[\left({\bf x}_{i} \cdot \rijUnit\right) - a_i\right] \left({\bf y}_{i} \cdot \rijUnit\right) 
-\mu \lambda \left[\left({\bf x}_{i} \cdot \riBUnit\right) - a_{bi}\right] \left({\bf y}_{i} \cdot \riBUnit\right)  \nonumber \\
v_i &= -\mu  (1-\lambda) \left[\left({\bf x}_{i} \cdot \rijUnit\right) - a_i\right] \left({\bf z}_{i} \cdot \rijUnit\right)   
-\mu \lambda \left[\left({\bf x}_{i} \cdot \riBUnit\right) - a_{bi}\right] \left({\bf z}_{i} \cdot \riBUnit\right)
\end{align}

% \begin{align}
% \label{eqn:controlLaws3dCBB}
% u_i &= -\lambda \mu \left(\mathbf{x}_i \cdot \frac{{\bf r}_{i}-{\bf r}_{T1}}{\left|{\bf r}_{i}-{\bf r}_{T1}\right|} - a_{T1}\right)\left(\mathbf{y}_i \cdot \frac{{\bf r}_{i}-{\bf r}_{T1}}{\left|{\bf r}_{i}-{\bf r}_{T1}\right|}\right)
% -(1-\lambda) \mu \left(\mathbf{x}_i \cdot \frac{{\bf r}_{i}-{\bf r}_{T2}}{\left|{\bf r}_{i}-{\bf r}_{T2}\right|} - a_{T2}\right)\left(\mathbf{y}_i \cdot \frac{{\bf r}_{i}-{\bf r}_{T2}}{\left|{\bf r}_{i}-{\bf r}_{T2}\right|}\right)
% , \nonumber \\
% v_i &= -\lambda \mu \left(\mathbf{x}_i \cdot \frac{{\bf r}_{i}-{\bf r}_{T1}}{\left|{\bf r}_{i}-{\bf r}_{T1}\right|} - a_{T1}\right)\left(\mathbf{z}_i \cdot \frac{{\bf r}_{i}-{\bf r}_{T1}}{\left|{\bf r}_{i}-{\bf r}_{T1}\right|}\right)
% -(1-\lambda) \mu \left(\mathbf{x}_i \cdot \frac{{\bf r}_{i}-{\bf r}_{T2}}{\left|{\bf r}_{i}-{\bf r}_{T2}\right|} - a_{T2}\right)\left(\mathbf{z}_i \cdot \frac{{\bf r}_{i}-{\bf r}_{T2}}{\left|{\bf r}_{i}-{\bf r}_{T2}\right|}\right),\; i=1,2.
% \end{align}
Note that the first term in each equation references the bearing to the other agent, and the second term references the bearing to the beacon\footnote{\kgEDIT{Note that the form of our control law can easily be extended to an arbitrary number of targets, with attentional weights that add up to $1$. Implementation of such a scheme would be limited only by sensor availability and computational power. An analysis of this generalized ``target-referenced'' control law will be the subject of future research.}}. In general, the neighbor-tracking goal may conflict with the beacon-referencing goal, i.e. there are no guarantees that both goals can be attained. Also, since $\lambda = 0$ and $\lambda = 1$ correspond to extreme cases of either exclusive tracking of the other agent or exclusive tracking of the beacon, we will assume $\lambda \in (0,1)$ in this work.

%\vspace{.3cm}

\noindent
{\bf Remark:} Our previous work \citep{KSG_BD_ACC_2018} introduced a beacon-referenced CB law similar to \eqref{eqn:beaconReferencedCBLaw} which used the full version of the CB pursuit law \citep{CP_3D_Kevin_2010,Galloway_PRS_16} with an added beacon-referencing term. In this work we have opted for the version given in \eqref{eqn:beaconReferencedCBLaw} for the following reasons. First, the version given in \eqref{eqn:beaconReferencedCBLaw} is simpler for an autonomous agent to implement in practice because it requires only bearing measurement. In particular, no range measurements or velocity estimations are required. Second, the version given in \eqref{eqn:beaconReferencedCBLaw} leads to a more tractable mathematical analysis. Lastly, simulations suggest that steady-state behaviors under \eqref{eqn:beaconReferencedCBLaw} and the control law given by \cite{KSG_BD_ACC_2018} are qualitatively similar.

%\vspace{.3cm}

In what follows, we will consider the various combinations presented by systems of 1-2 mobile agents employing the beacon-referenced control law \eqref{eqn:beaconReferencedCBLaw} with respect to 1-2 fixed beacons. For each case, we will analyze the closed-loop dynamics to determine existence conditions and steady-state characterization of relative equilibria (which correspond to circling motions). Our approach is to start with the lowest level of complexity in terms of the dimensionality of the system dynamics, and then progress to increasing levels of complexity. More specifically, we will consider the following three configurations (in order):
\begin{itemize}
    \item One mobile agent with two fixed beacons (Section \ref{sec:oneAgentTwoBeacons});
    \item Two mobile agents with one fixed beacon (Section \ref{sec:twoAgentsOneBeacon});
    \item Two mobile agents with two distinct fixed beacons (Section \ref{sec:twoAgentsTwoBeacons});
\end{itemize}

%
%
%
%%%%%%%%%%%%%%%%%%%%%%%%%%%%%%%%%%%%%%%%%%%%%%%%%%%%%%%%%%%%%%%%%%%%%%%%%%%%%%%%
\section{Configuration I: One Agent With Two Beacons}
\label{sec:oneAgentTwoBeacons}
%%%%%%%%%%%%%%%%%%%%%%%%%%%%%%%%%%%%%%%%%%%%%%%%%%%%%%%%%%%%%%%%%%%%%%%%%%%%%%%%
%
We start by considering the case of a single agent employing the beacon-referenced control law \eqref{eqn:beaconReferencedCBLaw} with respect to two fixed beacons. To this end, we let $a_{bi} \in [-1,1]$ represent the control parameter with reference to beacon $i$ (for $i=1,2$), and denote
\begin{align}
 &\bar{x}_{1bi} \triangleq {\bf x}_1 \cdot \rOneBiUnit, \quad 
\bar{y}_{1bi} \triangleq {\bf y}_i \cdot \rOneBiUnit, \quad
\bar{z}_{1bi} \triangleq {\bf z}_i \cdot \rOneBiUnit, \; i=1,2
% \nonumber \\
%  &\bar{x}_{1b2} \triangleq {\bf x}_1 \cdot \rOneBTwoUnit, \quad 
% \bar{y}_{1b2} \triangleq {\bf y}_i \cdot \rOneBTwoUnit, \quad
% \bar{z}_{1b2} \triangleq {\bf z}_i \cdot \rOneBTwoUnit
\end{align}
so that we can express a two-beacon version of \eqref{eqn:beaconReferencedCBLaw} as
\begin{align}
\label{eqn:beaconReferencedCBLawForTwoBeacons}
u_1 &= -\mu (1-\lambda) \left(\bar{x}_{1b2} - a_{b2}\right) \bar{y}_{1b2}
-\mu \lambda \left(\bar{x}_{1b1} - a_{b1}\right) \bar{y}_{1b1} \nonumber \\
v_1 &= -\mu (1-\lambda) \left(\bar{x}_{1b2} - a_{b2}\right) \bar{y}_{1b2}
-\mu \lambda \left(\bar{x}_{1b1} - a_{b1}\right) \bar{z}_{1b1}.
\end{align}
Since $\mathbf{x}_1$, $\mathbf{y}_1$, and $\mathbf{z}_1$ make up an orthonormal frame, we have the vector decomposition
\begin{equation}
    \rOneBiUnit = \bar{x}_{1bi} {\bf x}_1 + \bar{y}_{1bi} {\bf y}_1 + \bar{z}_{1bi} {\bf z}_1, \; i=1,2,
\end{equation}
and therefore substituting \eqref{eqn:beaconReferencedCBLawForTwoBeacons} into the first two equations from \eqref{Explicit_MODEL} yields
\begin{align}
\label{eqn:closedLoopEqnsOneAgentTwoBeacons}
\dot{\bf r}_1 &= \mathbf{x}_1, \nonumber \\
\dot{\bf x}_1 &= u_1 {\bf y}_1 +v_1 {\bf z}_1 \nonumber \\
&= 
-(1-\lambda)\mu (\bar{x}_{1b2} - a_{b2})\left(\bar{y}_{1b2}{\bf y}_1+\bar{z}_{1b2}{\bf z}_1\right)- \lambda \mu (\bar{x}_{1b1} - a_{b1})\left(\bar{y}_{1b1}{\bf y}_1+\bar{z}_{1b1}{\bf z}_1\right) \nonumber \\
&= 
-(1-\lambda)\mu (\bar{x}_{1b2} - a_{b2})\left(\rOneBTwoUnit - \bar{x}_{1b2}{\bf x}_1\right)- \lambda \mu (\bar{x}_{1b1} - a_{b1})\left(\rOneBUnit - \bar{x}_{1b1}{\bf x}_1\right).
\end{align}
Note that \eqref{eqn:closedLoopEqnsOneAgentTwoBeacons} is a self-contained set of dynamics \bdEDIT{and the complete $[\mathbf{x}_1,\mathbf{y}_1,\mathbf{z}_1]$ frame can be reconstructed from the evolution of this reduced dynamics \citep{Justh_PSK_3Dformation}}. If we restrict our analysis away from states corresponding to collocation of the agent with the beacons, we can denote
\begin{equation}
    \rho_{1b1} \triangleq \left|{\bf r}_{1b1}\right| > 0, \quad \rho_{1b2} \triangleq \left|{\bf r}_{1b2}\right| >0.
\end{equation}
\kgEDIT{Then, as is demonstrated in the electronic supplementary material, the corresponding shape dynamics (i.e. the relative configuration of the agent with respect to the beacons) can be expressed as}
\begin{align}
\label{eqn:shapeDynamicsForOneAgentTwoBeacons}
\dot{\bar{x}}_{1b1} 
&= -(1-\lambda)(\mu (\bar{x}_{1b2} - a_{b2}))\left(\rOneBUnit \cdot \rOneBTwoUnit - \bar{x}_{1b1}\bar{x}_{1b2} \right)  -\left(\lambda\mu(\bar{x}_{1b1}-a_{b1}) - \frac{1}{\rho_{1b1}}\right)\Bigl(1 - \bar{x}_{1b1}^2 \Bigr), \nonumber \\
\dot{\bar{x}}_{1b2} 
&= -\lambda\mu (\bar{x}_{1b1} - a_{b1}))\left(\rOneBUnit \cdot \rOneBTwoUnit - \bar{x}_{1b1}\bar{x}_{1b2} \right)  -\left((1-\lambda)\mu(\bar{x}_{1b2}-a_{b2}) - \frac{1}{\rho_{1b2}}\right)\Bigl(1 - \bar{x}_{1b2}^2 \Bigr), \nonumber \\
\dot{\rho}_{1b1} 
&= \bar{x}_{1b1} \nonumber \\
\dot{\rho}_{1b2} 
&= \bar{x}_{1b2},
\end{align}
where \eqref{eqn:closureConstraint} can be used to derive the Law of Cosines relationship
\begin{equation}
\label{eqn:lawOfCosinesOneAgentTwoBeacons}
  \rOneBUnit \cdot \rOneBTwoUnit = \frac{\rho_{1b1}^2 + \rho_{1b2}^2 -(2b)^2}{2\rho_{1b1}\rho_{1b2}}.
\end{equation}
Since the dot product of unit vectors must take values in $[-1,1]$, and the extreme values are not an option in this case because they correspond to configurations in which the agent and both beacons are collinear, \eqref{eqn:lawOfCosinesOneAgentTwoBeacons} implies that 
\begin{equation}
    -2\rho_{1b1}\rho_{1b2} < \rho_{1b1}^2 + \rho_{1b2}^2 -(2b)^2 < 2\rho_{1b1}\rho_{1b2},
\end{equation}
i.e. 
\begin{equation}
\label{eqn:rhoConstraintsOneAgentTwoBeacons}
\rho_{1b1} + \rho_{1b2} > 2b, \quad |\rho_{1b1} - \rho_{1b2}| < 2b.
\end{equation}

Our goal is to determine existence conditions and steady-state characterization for equilibria of the shape dynamics \eqref{eqn:shapeDynamicsForOneAgentTwoBeacons}, and therefore we set the shape dynamics to zero to obtain the necessary conditions $\bar{x}_{1b1}=0=\bar{x}_{1b2}$ and 
\begin{align}
\label{eqn:shapeDynamicsForOneAgentTwoBeaconsEquilibria}
&(1-\lambda)\mu a_{b2}\left(\frac{\rho_{1b1}^2 + \rho_{1b2}^2 -(2b)^2}{2\rho_{1b1}\rho_{1b2}}\right)  +\lambda\mu a_{b1} + \frac{1}{\rho_{1b1}}
= 0, \nonumber \\
&\lambda\mu a_{b1}\left(\frac{\rho_{1b1}^2 + \rho_{1b2}^2 -(2b)^2}{2\rho_{1b1}\rho_{1b2}}\right)  +(1-\lambda)\mu a_{b2} + \frac{1}{\rho_{1b2}}
=0.
\end{align}
Noting that the parameters show up together in the same patterns, we will denote
\begin{equation}
    \tilde{a}_{b1} \triangleq \lambda\mu a_{b1} \in \mathds{R}, \qquad \tilde{a}_{b2} \triangleq (1-\lambda)\mu a_{b2} \in \mathds{R}.
\end{equation}
Then setting each equation in \eqref{eqn:shapeDynamicsForOneAgentTwoBeaconsEquilibria} over a common denominator, we have
\begin{align}
\label{eqn:shapeDynamicsForOneAgentTwoBeaconsEquilibriaV2}
&\tilde{a}_{b2}\left(\rho_{1b1}^2 + \rho_{1b2}^2 -4b^2 \right)  +2\tilde{a}_{b1}\rho_{1b1}\rho_{1b2} + 2\rho_{1b2}
=0,
\nonumber \\
&\tilde{a}_{b1}\left(\rho_{1b1}^2 + \rho_{1b2}^2 -4b^2 \right)  +2\tilde{a}_{b2}\rho_{1b1}\rho_{1b2} + 2\rho_{1b1}
=0.
\end{align}

% %
% %

% \subsection{Special cases for configuration I}
% We'll proceed by considering a few of the special cases. First, \textbf{if $\mathbf{\tilde{a}_{b1}=0}$}, then \eqref{eqn:shapeDynamicsForOneAgentTwoBeaconsEquilibriaV2} simplifies to
% \begin{align}
% \label{eqn:shapeDynamicsForOneAgentTwoBeaconsEquilibria_ab1_zero}
% 0 
% &= \tilde{a}_{b2}\rho_{1b1}^2 + \tilde{a}_{b2}\rho_{1b2}^2 + 2\rho_{1b2} -4\tilde{a}_{b2}b^2, \nonumber \\
% %
% 0
% &= 2\tilde{a}_{b2}\rho_{1b1}\rho_{1b2} + 2\rho_{1b1},
% \end{align}
% and it follows immediately that we must have 
% \begin{equation}
%     \rho_{1b2} = \frac{1}{-\tilde{a}_{b2}} = \frac{1}{(1-\lambda)\mu(-a_{b2})}
% \end{equation}
% with $a_{b2} < 0$. Substituting back into the first equation from \eqref{eqn:shapeDynamicsForOneAgentTwoBeaconsEquilibria_ab1_zero}, we have
% \begin{align}
%     0 = \tilde{a}_{b2}\rho_{1b1}^2 + \frac{1}{\tilde{a}_{b2}} - \frac{2}{\tilde{a}_{b2}} -4\tilde{a}_{b2}b^2 
%     = \rho_{1b1}^2 - \frac{1}{\tilde{a}_{b2}^2} -4b^2,
% \end{align}
% and consequently we have
% \begin{equation}
%     \rho_{1b1}^2 = \frac{1}{\tilde{a}_{b2}^2} +4b^2 = \rho_{1b2}^2 + (2b)^2.
% \end{equation}
% So for the special case $\tilde{a}_{b1}=0$, a circling equilibrium exists if and only if $a_{b2} < 0$, with equilibrium values given by 
% \begin{equation}
%     \rho_{1b2} = \frac{1}{(1-\lambda)\mu(-a_{b2})}, \qquad \rho_{1b1} = \sqrt{\rho_{1b2}^2 + (2b)^2}.
% \end{equation}
% The equilibrium formation is a right triangle with the agent circling around beacon 2. A corresponding result holds if $a_{b2} = 0$ and $a_{b1} < 0$.

Before proceeding, we \textbf{consider a special case for which} $\mathbf{\tilde{a}_{b1} = \tilde{a}_{b2} \neq 0}$, i.e. $\lambda a_{b1} = (1-\lambda) a_{b2} \neq 0$. Substituting into \eqref{eqn:shapeDynamicsForOneAgentTwoBeaconsEquilibriaV2} and taking the difference of the two equations yields $\rho_{1b1} = \rho_{1b2}$. Substituting this constraint back into the first equation in \eqref{eqn:shapeDynamicsForOneAgentTwoBeaconsEquilibriaV2} yields
\begin{align}
\label{eqn:shapeDynamicsForOneAgentTwoBeaconsEquilibria_sameLengths}
% 0
% &= \tilde{a}_{b1}\rho_{1b1}^2 +2\tilde{a}_{b1}\rho_{1b1}^2 + \tilde{a}_{b1}\rho_{1b2}^2 + 2\rho_{1b1} -4\tilde{a}_{b1}b^2 \nonumber \\
%&=
4\tilde{a}_{b1}\left[\rho_{1b1}^2 + \left(\frac{1}{2\tilde{a}_{b1}}\right)\rho_{1b1} -b^2\right] = 0.
\end{align}
From our constraint \eqref{eqn:rhoConstraintsOneAgentTwoBeacons} we have $\rho_{1b1} + \rho_{1b2} > 2b$, which in this case requires $\rho_{1b1} > b$, and therefore it follows from \eqref{eqn:shapeDynamicsForOneAgentTwoBeaconsEquilibria_sameLengths} that we must have $\tilde{a}_{b1} < 0$. Thus we can summarize this case by stating that if $\tilde{a}_{b1} = \tilde{a}_{b2} < 0$, then an equilibrium exists with 
\begin{equation}
    \rho_{1b1} = \rho_{1b2} = \left(\frac{1}{2}\right)\left(-\frac{1}{2\tilde{a}_{b1}} - \sqrt{\frac{1}{4\tilde{a}_{b1}^2} + 4}\right)
    = \left(\frac{1}{-4\tilde{a}_{b1}} \right) \left(1 + \sqrt{1 + \left(4\tilde{a}_{b1}b\right)^2} \right)
\end{equation}

% \subsection{General case for configuration I}
\textbf{If} $\mathbf{\tilde{a}_{b1} \neq \tilde{a}_{b2}}$, then the two equations in  \eqref{eqn:shapeDynamicsForOneAgentTwoBeaconsEquilibriaV2} represent the intersection of two conic sections. While simple substitution leads to a 4th-order equation and difficult mathematical analysis, one can use methods such as conic pencils to determine the intersection points. However, in the present case a straightforward change of variables greatly simplifies our analysis. In particular, if we let $\rho_{1b+} = \rho_{1b2} + \rho_{1b1}$ and $\rho_{1b-} = \rho_{1b2} - \rho_{1b1}$, i.e. $\rho_{1b1} = (\rho_{1b+} - \rho_{1b-})/2$ and $\rho_{1b2} = (\rho_{1b+} + \rho_{1b-})/2$, then we can express \eqref{eqn:shapeDynamicsForOneAgentTwoBeaconsEquilibriaV2} (after some algebraic manipulation) as
\begin{align}
\label{eqn:shapeDynamicsForOneAgentTwoBeaconsEquilibriaChangeVars}
& \tilde{a}_{b2}\left(\frac{\rho_{1b+}^2 + \rho_{1b-}^2}{2} \right) +  \tilde{a}_{b1}\left(\frac{\rho_{1b+}^2 - \rho_{1b-}^2}{2} \right) +\left(\rho_{1b+} + \rho_{1b-} \right)  - 4\tilde{a}_{b2}b^2
=0, \nonumber \\
& \tilde{a}_{b1}\left(\frac{\rho_{1b+}^2 + \rho_{1b-}^2}{2} \right) +  \tilde{a}_{b2}\left(\frac{\rho_{1b+}^2 - \rho_{1b-}^2}{2} \right) +\left(\rho_{1b+} - \rho_{1b-} \right)  - 4\tilde{a}_{b1}b^2
=0.
\end{align}
Then taking the sum and difference of the two equations in \eqref{eqn:shapeDynamicsForOneAgentTwoBeaconsEquilibriaChangeVars}, we have
\begin{align}
\label{eqn:shapeDynamicsForOneAgentTwoBeaconsEquilibriaChangeVarsV2}
& \left(\tilde{a}_{b2} + \tilde{a}_{b1} \right)\rho_{1b+}^2 + 2\rho_{1b+} - 4\left(\tilde{a}_{b2} + \tilde{a}_{b1} \right)b^2
=0, \nonumber \\
& \left(\tilde{a}_{b2} - \tilde{a}_{b1} \right)\rho_{1b-}^2 + 2\rho_{1b-} - 4\left(\tilde{a}_{b2} - \tilde{a}_{b1} \right)b^2 
=0.
\end{align}
We note that if $\tilde{a}_{b2} + \tilde{a}_{b1} = 0$ then the first equation would require $\rho_{1b+} = 0$ which is not possible, and thus we assume $\tilde{a}_{b2} + \tilde{a}_{b1} \neq 0$. Since we've already assumed $\tilde{a}_{b2} - \tilde{a}_{b1} \neq 0$, we can express \eqref{eqn:shapeDynamicsForOneAgentTwoBeaconsEquilibriaChangeVarsV2} as 
\begin{align}
\label{eqn:shapeDynamicsForOneAgentTwoBeaconsEquilibriaChangeVarsV3a}
& \rho_{1b+}^2 + \left(\frac{2}{\tilde{a}_{b2} + \tilde{a}_{b1}}\right)\rho_{1b+} - (2b)^2 
=0,  \\
\label{eqn:shapeDynamicsForOneAgentTwoBeaconsEquilibriaChangeVarsV3b}
& \rho_{1b-}^2 + \left(\frac{2}{\tilde{a}_{b2} - \tilde{a}_{b1}}\right)\rho_{1b-} - (2b)^2
=0.
\end{align}
Since the first constraint in \eqref{eqn:rhoConstraintsOneAgentTwoBeacons} requires $\rho_{1b+} > 2b > 0$, i.e. $\rho_{1b+}^2 > (2b)^2$, it follows from \eqref{eqn:shapeDynamicsForOneAgentTwoBeaconsEquilibriaChangeVarsV3a} that we must have $\tilde{a}_{b2} + \tilde{a}_{b1} < 0$ with the equilibrium value
\begin{equation}
\label{eqn:rho1bPlusRoots}
    \rho_{1b+} = \left(\frac{1}{-(\tilde{a}_{b2} + \tilde{a}_{b1})} \right)\left(1 + \sqrt{1+\Bigl(2b(\tilde{a}_{b2} + \tilde{a}_{b1})\Bigr)^2}\right).
\end{equation}
Since the second constraint in \eqref{eqn:rhoConstraintsOneAgentTwoBeacons} requires $|\rho_{1b-}| < 2b$, i.e. $\rho_{1b-}^2 < (2b)^2$, it follows from \eqref{eqn:shapeDynamicsForOneAgentTwoBeaconsEquilibriaChangeVarsV3b} that we must have $(\tilde{a}_{b2} - \tilde{a}_{b1})\rho_{1b-} > 0$. Solving for the roots of \eqref{eqn:shapeDynamicsForOneAgentTwoBeaconsEquilibriaChangeVarsV3b} yields
\begin{align}
\label{eqn:rho1bMinusRoots}
    \rho_{1b-} 
    &= \frac{1}{2} \left[-\left(\frac{2}{\tilde{a}_{b2} - \tilde{a}_{b1}} \right) \pm \sqrt{\left(\frac{2}{\tilde{a}_{b2} - \tilde{a}_{b1}} \right)^2+4(2b)^2}\right] 
%\nonumber \\
% &= \frac{1}{2} \left[-\left(\frac{2}{\tilde{a}_{b2} - \tilde{a}_{b1}} \right) \pm \left(\frac{2}{|\tilde{a}_{b2} - \tilde{a}_{b1}|} \right)\sqrt{1+\Bigl(2b(\tilde{a}_{b2} - \tilde{a}_{b1})\Bigr)^2}\right] \nonumber \\
    = \left(\frac{1}{\tilde{a}_{b2} - \tilde{a}_{b1}} \right)\left(-1 \pm \sqrt{1+\Bigl(2b(\tilde{a}_{b2} - \tilde{a}_{b1})\Bigr)^2}\right),
\end{align}
and application of the constraint $(\tilde{a}_{b2} - \tilde{a}_{b1})\rho_{1b-} > 0$ requires that the second term in the \eqref{eqn:rho1bMinusRoots} product is positive. The following proposition summarizes the analysis of configuration 1 relative equilibria.

\begin{proposition}
Consider a system in which a single agent employs the beacon-referenced CB pursuit law \eqref{eqn:beaconReferencedCBLaw} with respect to two fixed beacons, according to the shape dynamics \eqref{eqn:shapeDynamicsForOneAgentTwoBeacons} parametrized by $\mu$, $\lambda$, and the CB parameters $a_{b1}$ and $a_{b2}$. Then a \textit{circling equilibrium} exists if and only if one of the following cases holds. (Note that in each case, at equilibrium we have $\bar{x}_{1b1}=0=\bar{x}_{1b2}$.)
%
%
% \\
% \indent (a) If $a_{b2}=0$ and $a_{b1} < 0$, then a circling equilibrium exists with 
% \begin{equation}
%     \rho_{1b1} = \frac{1}{\lambda\mu(-a_{b1})}, \qquad \rho_{1b2} = \sqrt{\rho_{1b1}^2 + (2b)^2}.
% \end{equation}
% %
% %
% \\
% \indent (b) If $a_{b1}=0$ and $a_{b2} < 0$, then a circling equilibrium exists with 
% \begin{equation}
%     \rho_{1b2} = \frac{1}{(1-\lambda)\mu(-a_{b2})}, \qquad \rho_{1b1} = \sqrt{\rho_{1b2}^2 + (2b)^2}.
% \end{equation}
% %
% %
\\
\indent (a) If $\lambda a_{b1} = (1-\lambda) a_{b2} < 0$, then a circling equilibrium exists with 
\begin{equation}
    \rho_{1b1} = \rho_{1b2} 
    = \left(\frac{1}{-4\mu\lambda a_{b1}} \right) \left(1 + \sqrt{1 + \left(4\mu\lambda a_{b1}b\right)^2} \right)
\end{equation}
\\
\indent (b) If $\lambda a_{b1} + (1-\lambda) a_{b2} < 0$, then a circling equilibrium exists with 
\begin{equation}
\label{eqn:oneAgentTwoBeaconsChangeOfVars}
   \rho_{1b1} = (\rho_{1b+} - \rho_{1b-})/2, \quad \rho_{1b2} = (\rho_{1b+} +
   \rho_{1b-})/2,
\end{equation} 
where
\begin{align}
\label{eqn:oneAgentTwoBeaconsCaseB}
    \rho_{1b+} = \frac{1 + \sqrt{1+\left[2b\mu\Bigl(\lambda a_{b1} + (1-\lambda) a_{b2}\Bigr)\right]^2}}{-\mu \bigl(\lambda a_{b1} + (1-\lambda) a_{b2} \bigr)},  \quad
     \rho_{1b-} = \frac{-1 + \sqrt{1+\left[2b\mu\Bigl(\lambda a_{b1} - (1-\lambda) a_{b2}\Bigr)\right]^2}}{-\mu \bigl(\lambda a_{b1} - (1-\lambda) a_{b2} \bigr)}.
\end{align}
\label{prop:existencePropSingleAgentTwoBeacons}
\end{proposition}
\begin{proof}
Follows from the discussion above.
%
% Cases (a)-(c) follow directly from the previous discussion. For case (d), we note that \eqref{eqn:rho1bPlusRoots} and \eqref{eqn:rho1bMinusRoots} imply that 
% \begin{align}
%     \rho_{1b+} = \frac{1 + \sqrt{1+\left[2b\Bigl(\lambda a_{b1} + (1-\lambda) a_{b2}\Bigr)\right]^2}}{-\lambda a_{b1}-(1-\lambda) a_{b2}},  \quad
%      \rho_{1b-} = \frac{-1 + \sqrt{1+\left[2b\Bigl(\lambda a_{b1} - (1-\lambda) a_{b2}\Bigr)\right]^2}}{-\lambda a_{b1}+(1-\lambda) a_{b2}}.
% \end{align}
% From there it follows by definition that
% \begin{align}
%     \rho_{1b1} = \frac{\left(-\lambda a_{b1}+(1-\lambda) a_{b2}\right)\left(1 + \sqrt{1+\left[2b\Bigl(\lambda a_{b1} + (1-\lambda) a_{b2}\Bigr)\right]^2}\right) - \left(-\lambda a_{b1}-(1-\lambda) a_{b2}\right)
%     }{2\left(\lambda^2 a_{b1}^2-(1-\lambda)^2 a_{b2}^2\right)}
% \end{align}
\end{proof}

%\vspace{.3cm}

\noindent
{\bf Remark:} \kgEDIT{The position of the agent with respect to the two beacons forms a triangle, and therefore triangle geometry can be used to show that the vertical positioning of the agent (i.e. the third component of $\mathbf{r}_{1}$, denoted as $r_{13}$) is given by $r_{13} = (\rho_{1b+}\rho_{1b-})/(4b) + b$ and the radius of the circle is given by $radius = \sqrt{\rho_{1b2}^2 - r^2_{13}}$.} Representative trajectories are shown in Figure \ref{fig:OneAgentTwoBeacon}. Note also that if one of the CB parameters is zero while the other is negative (e.g. if $a_{b2}=0$ and $a_{b1} < 0$), then the conditions of case (b) are satisfied, and simplifying \eqref{eqn:oneAgentTwoBeaconsChangeOfVars} and \eqref{eqn:oneAgentTwoBeaconsCaseB} yields
\begin{equation}
    \rho_{1b1} = \frac{1}{\lambda\mu(-a_{b1})}, \qquad \rho_{1b2} = \sqrt{\rho_{1b1}^2 + (2b)^2}.
\end{equation}

\begin{figure}
\centering
\begin{subfigure}{.5\textwidth}
  \centering
  \includegraphics[width=.8\linewidth]{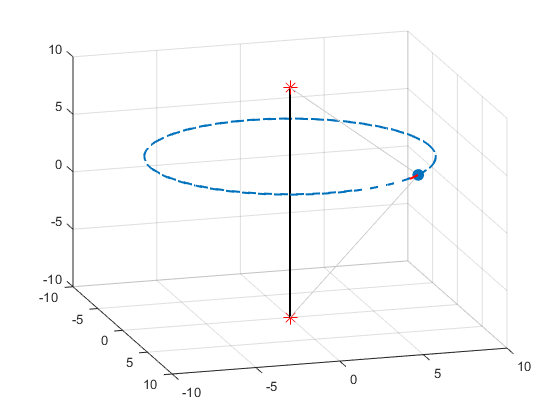}
%   \caption{A subfigure}
  \label{fig:sub1}
\end{subfigure}%
\begin{subfigure}{.5\textwidth}
  \centering
  \includegraphics[width=.8\linewidth]{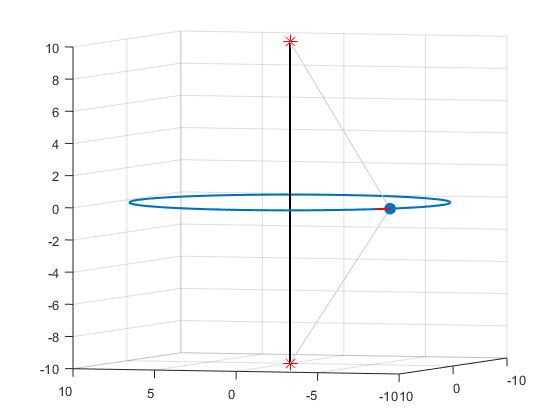}
%   \caption{A subfigure}
  \label{fig:sub2}
\end{subfigure}
\caption{These figures depict circling equilibria for configuration I (i.e. one agent referencing two fixed beacons) with beacon positioning $b=10$ and control parameters $a_0=-0.156$, $a = -0.187$, and $\lambda = 0.45$ in the simulation on the left, and $a_0=-0.156$, $a = -0.156$, and $\lambda = 0.5$ in the simulation on the right.}
\label{fig:OneAgentTwoBeacon}
\end{figure}
%
%
%
%%%%
%%%%

%%%%%%%%%%%%%%%%%%%%%%%%%%%%%%%%%%%%%%%%%%%%%%%%%%%%%%%%%%%%%%%%%%%%%%%%%%%%%%%%
\section{Shape Variable Representation for the Two-Agent Configurations}
\label{sec:ShapeVarDescription}
%%%%%%%%%%%%%%%%%%%%%%%%%%%%%%%%%%%%%%%%%%%%%%%%%%%%%%%%%%%%%%%%%%%%%%%%%%%%%%%%
%
In the remainder of this work, every configuration will involve at least one beacon and two mobile agents (engaged in a mutual pursuit). Moreover, to simplify analysis, we introduce the following assumptions\footnote{For future work, we intend to relax some of these assumptions to explore the broader space of possible system behaviors.}:
\begin{itemize}
\item[(A1)] The bearing offset parameters with respect to the beacon are common for both agents, i.e. $a_{b1} = a_{b2} = a_0$;
\item[(A2)] The bearing offset parameters with respect to the other agent are the same for both agents, i.e. $a_1 = a_2 = a$;
\end{itemize}
so that the closed-loop dynamics \bdEDIT{on a 10-dimensional \emph{reduced space} \citep{KSG_BD_ACC_2018}} are given by
\begin{align}
\label{eqn:closedLoopDynamics}
\dot{\bf r}_1 &= {\bf x}_1, \nonumber \\
\dot{\bf r}_2 &= {\bf x}_2, \nonumber \\
\dot{\bf x}_1 
&= 
-(1-\lambda)\mu (\bar{x}_1 - a)\left(\rUnit - \bar{x}_1{\bf x}_1\right)   
- \lambda \mu (\bar{x}_{1b1} - a_0)\left(\rOneBUnit - \bar{x}_{1b1} {\bf x}_1\right),
\nonumber \\
\dot{\bf x}_2
&= 
-(1-\lambda)\mu (\bar{x}_2 - a)\left(-\rUnit - \bar{x}_2{\bf x}_2\right) 
- \lambda \mu (\bar{x}_{2b2} - a_0)\left(\rTwoBUnit - \bar{x}_{2b2} {\bf x}_2\right),
\end{align}
where $\mathbf{r}=\mathbf{r}_1 - \mathbf{r}_2$ and 
\begin{equation}
    \bar{x}_{1} = \mathbf{x}_1 \cdot\frac{\mathbf{r}_{12}}{\left|{\mathbf{r}_{12}}\right|} =  \mathbf{x}_1 \cdot \rUnit, \quad
\bar{x}_{2} = \mathbf{x}_2 \cdot\frac{\mathbf{r}_{21}}{\left|{\mathbf{r}_{21}}\right|} =  \mathbf{x}_2 \cdot \left(-\rUnit\right).
\end{equation}
\bdEDIT{As the complete frame $[\mathbf{x}_i,\mathbf{y}_i,\mathbf{z}_i]$ can be reconstructed from the evolution of this reduced dynamics, we restrict our attention to \eqref{eqn:closedLoopDynamics} in the remainder of this work.}
%\bdEDIT{Note that the closed-loop dynamics \eqref{eqn:closedLoopDynamics} evolve on a 10-dimensional \emph{reduced space} \citep{Justh_PSK_3Dformation,KSG_BD_ACC_2018}}. 
As a step towards developing a scalar parameterization of the shape dynamics (i.e. the dynamics of the relative configuration of the agents and beacons), we also define 
\begin{equation}
    \tilde{x} = \mathbf{x}_1 \cdot \mathbf{x}_2, \quad 
\rho = |\mathbf{r}| >0.
\end{equation}
This allows us to express the dynamics of these variables as\footnote{Please refer the the electronic supplementary material for detailed derivation of \eqref{eqn:2AgentShapeDynamicsSet1}.}: 
%Then one can show (see the electronic supplementary material) that the derivatives of these scalar shape variables are given by 
\begin{align}
\label{eqn:2AgentShapeDynamicsSet1}
\dot{\bar{x}}_1
&= -(1-\lambda)\mu (\bar{x}_1 - a) \left(1 - \bar{x}_1^2 \right)  -\lambda\mu(\bar{x}_{1b1}-a_0)\left(\rOneBUnit \cdot \rUnit - \bar{x}_{1b1}\bar{x}_1 \right)
+ \frac{1}{\rho} \Bigl(1 - \tilde{x} - \bar{x}_1^2 - \bar{x}_1 \bar{x}_{2} \Bigr),  \nonumber \\
\dot{\bar{x}}_2
&= -(1-\lambda)\mu (\bar{x}_2 - a) \left(1 - \bar{x}_2^2 \right)  -\lambda\mu(\bar{x}_{2b2}-a_0)\left(-\rTwoBUnit \cdot \rUnit - \bar{x}_{2b2}\bar{x}_2 \right)
+ \frac{1}{\rho} \Bigl(1 - \tilde{x} - \bar{x}_2^2 - \bar{x}_1 \bar{x}_{2} \Bigr),  \nonumber \\
\dot{\bar{x}}_{1b1} 
&= -(1-\lambda)(\mu (\bar{x}_1 - a))\left(\rOneBUnit \cdot \rUnit - \bar{x}_{1b1}\bar{x}_1 \right)  -\left(\lambda\mu(\bar{x}_{1b1}-a_0) - \frac{1}{\rho_{1b1}}\right)\Bigl(1 - \bar{x}_{1b1}^2 \Bigr), \nonumber \\
\dot{\bar{x}}_{2b2} 
&= -(1-\lambda)(\mu (\bar{x}_2 - a))\left(-\rTwoBUnit \cdot \rUnit - \bar{x}_{2b2}\bar{x}_2 \right)  -\left(\lambda\mu(\bar{x}_{2b2}-a_0) - \frac{1}{\rho_{2b2}}\right)\Bigl(1 - \bar{x}_{2b2}^2 \Bigr), \nonumber \\
\dot{\tilde{x}}
&= (1-\lambda)\mu (\bar{x}_1 - a) \left(\bar{x}_{2} + \tilde{x} \bar{x}_1\right)  + (1-\lambda)\mu (\bar{x}_2 - a) \left(\bar{x}_{1} + \tilde{x} \bar{x}_2\right)  \nonumber \\
&\qquad -\lambda\mu(\bar{x}_{1b1}-a_0)\left({\bf x}_2 \cdot \rOneBUnit  - \bar{x}_{1b1}\tilde{x}\right) 
-\lambda\mu(\bar{x}_{2b2}-a_0)\left({\bf x}_1 \cdot \rTwoBUnit  - \bar{x}_{2b2}\tilde{x}\right) 
\nonumber \\
\dot{\rho} 
&= \bar{x}_1 + \bar{x}_{2}, \nonumber \\
\dot{\rho}_{1b1} 
&= \bar{x}_{1b1} \nonumber \\
\dot{\rho}_{2b2} 
&= \bar{x}_{2b2},
\end{align}
where we note that the explicit dot product terms still need to be represented in terms of the shape variables to make this completely self-contained. We can derive expressions for these dot product terms by making use of our vector closure constraint \eqref{eqn:closureConstraint}. For example, \eqref{eqn:closureConstraint} implies that
\begin{align}
    {\bf x}_2 \cdot \mathbf{r}_{1b1} 
    % &
    = {\bf x}_2 \cdot \left(\hat{\bf b} +  \mathbf{r}_{2b2} + \mathbf{r} \right)
    % \nonumber \\
    % &= {\bf x}_2 \cdot \hat{\bf b} + \rho_{2b2} \left({\bf x}_2 \cdot \rTwoBUnit\right) +\rho\left({\bf x}_2 \cdot \rUnit\right) \nonumber \\
    % &
    = {\bf x}_2 \cdot \hat{\bf b} + \rho_{2b2} \bar{x}_{2b2} -\rho \bar{x}_2,
\end{align}
from which it follows that 
\begin{align}
\label{eqn:x2r1b1eqnOrig}
    {\bf x}_2 \cdot \rOneBUnit 
    &= \frac{1}{\rho_{1b1}}\left({\bf x}_2 \cdot \hat{\bf b} + \rho_{2b2} \bar{x}_{2b2} -\rho \bar{x}_2\right) 
    \in [-1,1].
\end{align}
Similar calculations lead to 
\begin{align}
\label{eqn:x1r2b2eqnOrig}
    {\bf x}_1 \cdot \rTwoBUnit 
    &= \frac{1}{\rho_{2b2}}\left(-{\bf x}_1 \cdot \hat{\bf b} + \rho_{1b1} \bar{x}_{1b1} -\rho \bar{x}_1\right) \in [-1,1],  \\
\label{eqn:rUnitb1Unit}
\rOneBUnit \cdot \rUnit 
    &= \frac{1}{2\rho\rho_{1b1}}\left(\rho_{1b1}^2 + \rho^2 - \rho_{2b2}^2 -2(\mathbf{r}_{2} \cdot \hat{\bf b}) \right) \in [-1,1], \\
\label{eqn:rUnitb2Unit}
    \rTwoBUnit \cdot \rUnit 
    &= \frac{1}{2\rho\rho_{2b2}}\left(\rho_{1b1}^2 - \rho^2 - \rho_{2b2}^2 -2(\mathbf{r}_{1} \cdot \hat{\bf b})\right) \in [-1,1].
\end{align}
From these calculations it is clear that we can make the shape dynamics \eqref{eqn:2AgentShapeDynamicsSet1} self-contained if we augment them with the variables
\begin{equation}
     \hat{r}_{1} = \mathbf{r}_{1} \cdot \hat{\bf b}, \quad 
     \hat{r}_{2} = \mathbf{r}_{2} \cdot \hat{\bf b}, \quad
     \hat{x}_{1} = \mathbf{x}_{1} \cdot \hat{\bf b}, \quad 
     \hat{x}_{2} = \mathbf{x}_{2} \cdot \hat{\bf b},
\end{equation}
with dynamics given by 
\begin{align}
\label{eqn:2AgentShapeDynamicsSet2}
    \dot{\hat{r}}_1 &= \hat{x}_1, \qquad
    \dot{\hat{r}}_2 = \hat{x}_2, \nonumber \\
%%%%
    \dot{\hat{x}}_1 
    &= -(1-\lambda)\mu (\bar{x}_1 - a)\left(\rUnit \cdot \hat{\bf b}  - \bar{x}_1{\bf x}_1 \cdot \hat{\bf b}\right)  
- \lambda \mu (\bar{x}_{1b1} - a_0)\left(\rOneBUnit \cdot \hat{\bf b} - \bar{x}_{1b1} {\bf x}_1 \cdot \hat{\bf b}\right) \nonumber \\
    &= -\frac{(1-\lambda)\mu (\bar{x}_1 - a)}{\rho}\left(\mathbf{r} \cdot \hat{\bf b} \right)  
- \frac{\lambda \mu (\bar{x}_{1b1} - a_0)}{\rho_{1b1}}\left(\mathbf{r}_{1b1} \cdot \hat{\bf b}\right)
+\mu\Bigl((1-\lambda) \bar{x}_1(\bar{x}_1 - a)+ \lambda \bar{x}_{1b1}(\bar{x}_{1b1} - a_0)\Bigr)\left( {\bf x}_1 \cdot \hat{\bf b}\right)
\nonumber \\
%     &= -\mu\left(\frac{(1-\lambda) (\bar{x}_1 - a)}{\rho} + \frac{\lambda (\bar{x}_{1b1} - a_0)}{\rho_{1b1}}\right)\hat{r}_1  
% +\left(\frac{(1-\lambda)\mu (\bar{x}_1 - a)}{\rho}\right)\hat{r}_2
% - \frac{\lambda \mu (\bar{x}_{1b1} - a_0)}{\rho_{1b1}}\left(-\mathbf{r}_{b1} \cdot \hat{\bf b}\right) 
% \nonumber \\
% &\qquad+\mu\Bigl((1-\lambda) \bar{x}_1(\bar{x}_1 - a)+ \lambda \bar{x}_{1b1}(\bar{x}_{1b1} - a_0)\Bigr)\hat{x}_1
% \nonumber \\
    &= -\mu\left(\frac{(1-\lambda) (\bar{x}_1 - a)}{\rho} + \frac{\lambda (\bar{x}_{1b1} - a_0)}{\rho_{1b1}}\right)\hat{r}_1  
+\left(\frac{(1-\lambda)\mu (\bar{x}_1 - a)}{\rho}\right)\hat{r}_2
- \frac{2\lambda \mu (\bar{x}_{1b1} - a_0)b^2}{\rho_{1b1}} 
\nonumber \\
&\qquad+\mu\Bigl((1-\lambda) \bar{x}_1(\bar{x}_1 - a)+ \lambda \bar{x}_{1b1}(\bar{x}_{1b1} - a_0)\Bigr)\hat{x}_1
\nonumber \\
%%%%
    \dot{\hat{x}}_2 
        &= -(1-\lambda)\mu (\bar{x}_2 - a)\left(-\rUnit \cdot \hat{\bf b}  - \bar{x}_2{\bf x}_2 \cdot \hat{\bf b}\right)  
- \lambda \mu (\bar{x}_{2b2} - a_0)\left(\rTwoBUnit \cdot \hat{\bf b} - \bar{x}_{2b2} {\bf x}_2 \cdot \hat{\bf b}\right) \nonumber \\
    &= -\frac{(1-\lambda)\mu (\bar{x}_2 - a)}{\rho}\left(-\mathbf{r} \cdot \hat{\bf b} \right)  
- \frac{\lambda \mu (\bar{x}_{2b2} - a_0)}{\rho_{2b2}}\left(\mathbf{r}_{2b2} \cdot \hat{\bf b}\right)
+\mu\Bigl((1-\lambda) \bar{x}_2(\bar{x}_2 - a)+ \lambda \bar{x}_{2b2}(\bar{x}_{2b2} - a_0)\Bigr)\left( {\bf x}_2 \cdot \hat{\bf b}\right)
\nonumber \\
%     &= -\mu\left(\frac{(1-\lambda) (\bar{x}_2 - a)}{\rho} + \frac{\lambda (\bar{x}_{2b2} - a_0)}{\rho_{2b2}}\right)\hat{r}_2  
% +\left(\frac{(1-\lambda)\mu (\bar{x}_2 - a)}{\rho}\right)\hat{r}_1
% - \frac{\lambda \mu (\bar{x}_{2b2} - a_0)}{\rho_{2b2}}\left(-\mathbf{r}_{b2} \cdot \hat{\bf b}\right) 
% \nonumber \\
% &\qquad+\mu\Bigl((1-\lambda) \bar{x}_2(\bar{x}_2 - a)+ \lambda \bar{x}_{1b1}(\bar{x}_{2b2} - a_0)\Bigr)\hat{x}_2
% \nonumber \\
    &= -\mu\left(\frac{(1-\lambda) (\bar{x}_2 - a)}{\rho} + \frac{\lambda (\bar{x}_{2b2} - a_0)}{\rho_{2b2}}\right)\hat{r}_2  
+\left(\frac{(1-\lambda)\mu (\bar{x}_2 - a)}{\rho}\right)\hat{r}_1
+ \frac{2\lambda \mu (\bar{x}_{2b2} - a_0)b^2}{\rho_{2b2}} 
\nonumber \\
&\qquad+\mu\Bigl((1-\lambda) \bar{x}_2(\bar{x}_2 - a)+ \lambda \bar{x}_{2b2}(\bar{x}_{2b2} - a_0)\Bigr)\hat{x}_2.
\end{align}
In the next sections we will characterize equilibria for these shape dynamics \eqref{eqn:2AgentShapeDynamicsSet1}, \eqref{eqn:2AgentShapeDynamicsSet2} and determine conditions under which those equilibria exist.
%%%%%%%%%%%%%%%%%%%%%%%%%%%%%%%%%%%%%%%%%%%%%%%%%%%%%%%%%%%%%%%%%%%%%%%%%%%%%%%%
\section{Configuration II: Two Agents Referencing A Single Beacon}
\label{sec:twoAgentsOneBeacon}
%%%%%%%%%%%%%%%%%%%%%%%%%%%%%%%%%%%%%%%%%%%%%%%%%%%%%%%%%%%%%%%%%%%%%%%%%%%%%%%%
%
We first consider the case in which there is a single beacon (at the origin) referenced by two agents in mutual pursuit\footnote{This is the case that was considered exclusively by \cite{KSG_BD_ACC_2018}, but with a different version of the beacon-referenced control law. As such, many of the results in this section are similar in spirit to those presented by \cite{KSG_BD_ACC_2018}.}. In this case $\hat{\bf b}$ is the zero vector, i.e. $\hat{r}_1= 0 = \hat{r}_2$ and $\hat{x}_1= 0 = \hat{x}_2$, and therefore our shape dynamics simplify to \eqref{eqn:2AgentShapeDynamicsSet1} with constraints given by \eqref{eqn:x2r1b1eqnOrig}, \eqref{eqn:x1r2b2eqnOrig}, \eqref{eqn:rUnitb1Unit}, and \eqref{eqn:rUnitb2Unit}. Then setting $\dot{\rho}=0$ and $\dot{\rho}_{1b1} = 0 = \dot{\rho}_{2b2}$ requires $\bar{x}_2 = -\bar{x}_1$ and $\bar{x}_{1b} = 0 = \bar{x}_{2b}$, and the shape dynamics on these nullclines are given by
\begin{align}
\label{eqn:2AgentsOneBeaconDynamicsAtEquilibrium}
\dot{\bar{x}}_1
&= -(1-\lambda)\mu (\bar{x}_1 - a) \left(1 - \bar{x}_1^2 \right)  
+ \lambda\mu a_0\left(\frac{\rho_{1b1}^2 +\rho^2 - \rho_{2b2}^2 }{2\rho \rho_{1b1}} \right) + \frac{1}{\rho} \Bigl(1 - \tilde{x}\Bigr),  \nonumber \\
\dot{\bar{x}}_2
&= -(1-\lambda)\mu (-\bar{x}_1 - a) \left(1 - \bar{x}_1^2 \right)  + \lambda\mu a_0\left(\frac{-\rho_{1b1}^2 +\rho^2 + \rho_{2b2}^2 }{2\rho \rho_{2b2}} \right) + \frac{1}{\rho} \Bigl(1 - \tilde{x}\Bigr),  \nonumber \\
\dot{\bar{x}}_{1b1} 
&= -(1-\lambda)\Bigl(\mu (\bar{x}_1 - a)\Bigr)\left(\frac{\rho_{1b1}^2 +\rho^2 - \rho_{2b2}^2 }{2\rho \rho_{1b1}} \right)  + \frac{1}{\rho_{1b1}} + \lambda\mu a_0, \nonumber \\
\dot{\bar{x}}_{2b2} 
&= -(1-\lambda)\Bigl(\mu (-\bar{x}_1 - a)\Bigr)\left(\frac{-\rho_{1b1}^2 +\rho^2 + \rho_{2b2}^2 }{2\rho \rho_{2b2}} \right) + \frac{1}{\rho_{2b2}} + \lambda\mu a_0, \nonumber \\
\dot{\tilde{x}}
&= \mu \bar{x}_{1} \left( \lambda a_0  \rho \left(\frac{\rho_{2b2}-\rho_{1b1}}{\rho_{1b1}\rho_{2b2}}\right) 
- 2(1-\lambda) \bar{x}_1\left(1 - \tilde{x}\right)\right).  
\end{align}
%%%%
%%%%
\subsection{Special case for configuration II: $a_0=0$}
If $a_0=0$, then \eqref{eqn:2AgentsOneBeaconDynamicsAtEquilibrium} simplifies to 
\begin{align}
\label{eqn:2AgentsOneBeaconA0DynamicsAtEquilibrium}
\dot{\bar{x}}_1
&= -(1-\lambda)\mu (\bar{x}_1 - a) \left(1 - \bar{x}_1^2 \right)  
 + \frac{1}{\rho} \Bigl(1 - \tilde{x}\Bigr),  \nonumber \\
\dot{\bar{x}}_2
&= -(1-\lambda)\mu (-\bar{x}_1 - a) \left(1 - \bar{x}_1^2 \right)  + \frac{1}{\rho} \Bigl(1 - \tilde{x}\Bigr),  \nonumber \\
\dot{\bar{x}}_{1b1} 
&= -(1-\lambda)\Bigl(\mu (\bar{x}_1 - a)\Bigr)\left(\frac{\rho_{1b1}^2 +\rho^2 - \rho_{2b2}^2 }{2\rho \rho_{1b1}} \right)  + \frac{1}{\rho_{1b1}} , \nonumber \\
\dot{\bar{x}}_{2b2} 
&= -(1-\lambda)\Bigl(\mu (-\bar{x}_1 - a)\Bigr)\left(\frac{-\rho_{1b1}^2 +\rho^2 + \rho_{2b2}^2 }{2\rho \rho_{2b2}} \right) + \frac{1}{\rho_{2b2}}, \nonumber \\
\dot{\tilde{x}}
&= - 2\mu(1-\lambda) \bar{x}_1^2\left(1 - \tilde{x}\right).  
\end{align}
We note that the first two equations can both be zero if and only if $\bar{x}_1 = \pm 1 $ or $\bar{x}_1 = 0$, and since the latter option is the only one that results in $\dot{\tilde{x}}=0$, we conclude that $\bar{x}_1=0$ at equilibrium. From this we can state the following proposition. 

\begin{proposition}
Consider a beacon-referenced mutual CB pursuit system with shape dynamics \eqref{eqn:2AgentsOneBeaconDynamicsAtEquilibrium} parametrized by $\mu$, $\lambda$, and the CB parameters $a$, $a_0$, with $a_0=0$. Then, a \textit{circling equilibrium} exists if and only if $a<0$, and the corresponding equilibrium values satisfy
\begin{equation}
\begin{aligned}
\bullet \quad &
\bar{x}_1 = \bar{x}_2 = 0,
\quad 
\bar{x}_{1b1} = \bar{x}_{2b2} = 0,
\quad 
\tilde{x} = -1,
\\
\bullet \quad &
\rho = \frac{2}{(1-\lambda)\mu (-a)},
\quad
\rho_{1b1} = \rho_{2b2}.
\end{aligned}
\label{eqn:a0SameAlphaCircValues}
\end{equation}
\label{prop:existenceProp_1}
\end{proposition}
\begin{proof}
Substituting $\bar{x}_1=0$ into \eqref{eqn:2AgentsOneBeaconA0DynamicsAtEquilibrium} and setting to zero leads to 
\begin{subequations}
\begin{align}
(1-\lambda)\mu a  + \frac{1}{\rho} \Bigl(1 - \tilde{x}\Bigr)
&=
0,
\label{eqn:2AgentDynamicsAtEquilibriumCommonCBParameter_a0Zero__1}
\\
(1-\lambda) \mu a\left(\frac{\rho_{1b1}^2 +\rho^2 - \rho_{2b2}^2 }{2\rho \rho_{1b1}} \right)
&=
\frac{-1}{\rho_{1b1}}, 
\label{eqn:2AgentDynamicsAtEquilibriumCommonCBParameter_a0Zero__2}
\\
(1-\lambda) \mu a\left(\frac{\rho_{2b2}^2 +\rho^2 - \rho_{1b1}^2 }{2\rho \rho_{2b2}} \right)
&=
\frac{-1}{\rho_{2b2}}.
\label{eqn:2AgentDynamicsAtEquilibriumCommonCBParameter_a0Zero__3}
\end{align}
\end{subequations}
Now, if $\tilde{x} = 1$, the first condition \eqref{eqn:2AgentDynamicsAtEquilibriumCommonCBParameter_a0Zero__1} holds true if and only if $a=0$. But with these choices for $\tilde{x}$ and $a$, the last two conditions \eqref{eqn:2AgentDynamicsAtEquilibriumCommonCBParameter_a0Zero__2}-\eqref{eqn:2AgentDynamicsAtEquilibriumCommonCBParameter_a0Zero__3} lead to 
$\frac{1}{\rho_{1b}} = \frac{1}{\rho_{2b}} = 0$, which cannot be true since both $\rho_{1b}$ and $\rho_{2b}$ are finite. Therefore we must have $\tilde{x} \neq 1$ at an equilibrium, and then the first condition \eqref{eqn:2AgentDynamicsAtEquilibriumCommonCBParameter_a0Zero__1} yields the equilibrium value of $\rho$ as
\begin{equation}
\rho = \frac{(1 - \tilde{x})}{(1-\lambda)\mu (-a)}.
\label{eqn:rhoEquilibrium_a0zero}
\end{equation}
As $\rho$ must be positive and finite, \eqref{eqn:rhoEquilibrium_a0zero} yields a meaningful solution if and only if $a < 0$. Substituting this solution for $\rho$ into \eqref{eqn:2AgentDynamicsAtEquilibriumCommonCBParameter_a0Zero__2}-\eqref{eqn:2AgentDynamicsAtEquilibriumCommonCBParameter_a0Zero__3}, we have 
\begin{subequations}
\begin{align}
\frac{1}{2 \rho_{1b1}}\biggl[-\left(1 - \tilde{x}	\right)\left(\frac{\rho_{1b1}^2 + \rho^{2} - \rho_{2b2}^2 }{\rho^2} \right) + 2 \biggr]
&=0,
\label{eqn:EQ_Final_Conditions_on_Nullcline__1}
\\
\frac{1}{2 \rho_{2b2}}\biggl[-\left(1 - \tilde{x}	\right)\left(\frac{-\rho_{1b1}^2 + \rho^{2} + \rho_{2b2}^2 }{\rho^2} \right) + 2 \biggr]
&= 0,  
\label{eqn:EQ_Final_Conditions_on_Nullcline__2}
\end{align}
\end{subequations}
which holds true if and only if $\rho_{1b1} = \rho_{2b2}$ and $\tilde{x} = -1$. 
\end{proof}

\begin{figure}
\begin{center}
  \includegraphics[width=0.4\textwidth]{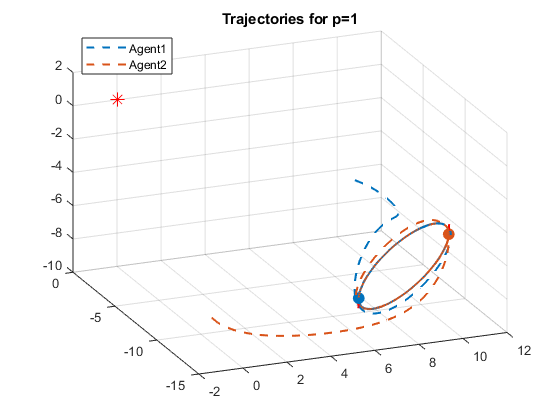}
  \caption{Illustration of the type of circling equilibrium described in Prop.~\ref{prop:existenceProp_1} ($\lambda=0.5$, $a=-0.7071$, $a_0=0$). The agents follow a circling trajectory on a plane perpendicular to the axis which passes through the beacon (denoted by the asterisk).}
  \vspace{-1.5em}
  \label{Fig:OffsetCircle}
\end{center}
\end{figure}

Figure~\ref{Fig:OffsetCircle} illustrates the type of circling equilibrium which is described in \textbf{Proposition \ref{prop:existenceProp_1}}. Note that the distance of each agent from the beacon (i.e. $\rho_{1b1}$) is determined by initial conditions, but the separation between the agents (i.e. $\rho$) is determined by the control parameters according to \eqref{eqn:a0SameAlphaCircValues}.

\subsection{General case for configuration II: $a_0 \neq 0$}
%%%%
We now shift our attention to the case where $a_0 \neq 0$, and show that circling equilibria exist in this scenario as well.
\begin{proposition}
Consider a beacon-referenced mutual CB pursuit system with shape dynamics \eqref{eqn:2AgentsOneBeaconDynamicsAtEquilibrium} parametrized by $\mu$, $\lambda$, and the CB parameters $a$, $a_0$, with $a_0\neq0$. The following statements are true.
\\
\indent (a) Whenever $(1-\lambda) a + \lambda a_0 < 0$, a circling equilibrium exists, with corresponding equilibrium values given by
\begin{equation}
\begin{aligned}
\bullet \quad &
\bar{x}_1 = \bar{x}_2 = 0,
\quad 
\bar{x}_{1b1} = \bar{x}_{2b2} = 0,
\quad 
\tilde{x} = -1,
\\
\bullet \quad &
\rho_{1b1} = \rho_{2b2} 
= \frac{1}{-\mu \bigl((1-\lambda)a + \lambda a_0\bigr)},
\\
\bullet \quad &
\rho 
= 2\rho_{1b1} = \frac{2}{-\mu \bigl((1-\lambda)a + \lambda a_0\bigr)}.
\end{aligned}
\end{equation}
\\
\indent (b) Whenever $(1-\lambda) a + \lambda a_0 < 0$, $a_0 < 0$, and $a>0$, a circling equilibrium exists, with corresponding equilibrium values given by
\begin{equation}
\begin{aligned}
\bullet \quad &
\bar{x}_1 = \bar{x}_2 = 0,
\quad 
\bar{x}_{1b1} = \bar{x}_{2b2} = 0,
\quad 
\tilde{x} = 1,
\\
\bullet \quad &
\rho_{1b1} = \rho_{2b2}
= \frac{\lambda a_0}{\mu\Bigl((1-\lambda)^2 a^2 - \lambda^2 a_0^2\Bigr)},
\\
\bullet \quad &
\rho 
= \frac{-2(1-\lambda) a}{\mu\Bigl((1-\lambda)^2 a^2 - \lambda^2 a_0^2\Bigr)}.
\end{aligned}
\end{equation}
\label{prop:existenceProp_TwoAgentsOneBeacon_a0_nonzero}
\end{proposition}
\begin{proof}
See the electronic supplementary material.
\end{proof}
%%%%%%%%%%%%%%%%%%%%%%%%%%%%%%%%%%%%%%%%%%%%%%%%%%%%%%%%%%%%%%%%%%%%%%%%%%%%%%%%
\section{Configuration III: Two Agents Referencing Two Distinct Beacons}
\label{sec:twoAgentsTwoBeacons}
%%%%%%%%%%%%%%%%%%%%%%%%%%%%%%%%%%%%%%%%%%%%%%%%%%%%%%%%%%%%%%%%%%%%%%%%%%%%%%%%
%
We now consider the most general case, for which two agents in mutual pursuit each reference a different beacon. For this case, our shape dynamics are given by \eqref{eqn:2AgentShapeDynamicsSet1}, \eqref{eqn:2AgentShapeDynamicsSet2} with constraints given by \eqref{eqn:x2r1b1eqnOrig}, \eqref{eqn:x1r2b2eqnOrig}, \eqref{eqn:rUnitb1Unit}, and \eqref{eqn:rUnitb2Unit}. As in previous cases, setting $\dot{\rho}=0$ and $\dot{\rho}_{1b1} = 0 = \dot{\rho}_{2b2}$ requires $\bar{x}_2 = -\bar{x}_1$ and $\bar{x}_{1b} = 0 = \bar{x}_{2b}$. If we also set $\dot{\hat{r}}_1 = 0 = \dot{\hat{r}}_2$, i.e. $\hat{x}_{1} = 0 = \hat{x}_{2}$, then our shape dynamics \eqref{eqn:2AgentShapeDynamicsSet1}-\eqref{eqn:2AgentShapeDynamicsSet2} along these nullclines are given by
\begin{align}
\label{eqn:2Agents2BeaconsDynamicsAtEquilibrium}
\dot{\bar{x}}_1
&= -(1-\lambda)\mu (\bar{x}_1 - a) \left(1 - \bar{x}_1^2 \right)  
+ \lambda\mu a_0\left(\frac{\rho_{1b1}^2 +\rho^2 - \rho_{2b2}^2 -2\hat{r}_2 }{2\rho \rho_{1b1}} \right) + \frac{1}{\rho} \Bigl(1 - \tilde{x}\Bigr),  \nonumber \\
\dot{\bar{x}}_2
&= -(1-\lambda)\mu (-\bar{x}_1 - a) \left(1 - \bar{x}_1^2 \right)  + \lambda\mu a_0\left(\frac{-\rho_{1b1}^2 +\rho^2 + \rho_{2b2}^2 +2\hat{r}_1 }{2\rho \rho_{2b2}} \right) + \frac{1}{\rho} \Bigl(1 - \tilde{x}\Bigr),  \nonumber \\
\dot{\bar{x}}_{1b1} 
&= -(1-\lambda)\Bigl(\mu (\bar{x}_1 - a)\Bigr)\left(\frac{\rho_{1b1}^2 +\rho^2 - \rho_{2b2}^2  -2\hat{r}_2}{2\rho \rho_{1b1}} \right)  + \frac{1}{\rho_{1b1}} + \lambda\mu a_0, \nonumber \\
\dot{\bar{x}}_{2b2} 
&= -(1-\lambda)\Bigl(\mu (-\bar{x}_1 - a)\Bigr)\left(\frac{-\rho_{1b1}^2 +\rho^2 + \rho_{2b2}^2  +2\hat{r}_1}{2\rho \rho_{2b2}} \right) + \frac{1}{\rho_{2b2}} + \lambda\mu a_0, \nonumber \\
\dot{\tilde{x}}
&= \mu \bar{x}_{1} \left( \lambda a_0  \rho \left(\frac{\rho_{2b2}-\rho_{1b1}}{\rho_{1b1}\rho_{2b2}}\right) 
- 2(1-\lambda) \bar{x}_1\left(1 - \tilde{x}\right)\right),
\nonumber \\
%
%%%%
\dot{\hat{x}}_1 
&= -\mu\left(\frac{(1-\lambda) (\bar{x}_1 - a)}{\rho} - \frac{\lambda  a_0}{\rho_{1b1}}\right)\hat{r}_1  
+\left(\frac{(1-\lambda)\mu (\bar{x}_1 - a)}{\rho}\right)\hat{r}_2
+ \frac{2\lambda \mu a_0 b^2}{\rho_{1b1}} 
\nonumber \\
%%%%
\dot{\hat{x}}_2 
&= -\mu\left(\frac{(1-\lambda) (-\bar{x}_1 - a)}{\rho} - \frac{\lambda a_0}{\rho_{2b2}}\right)\hat{r}_2  
+\left(\frac{(1-\lambda)\mu (-\bar{x}_1 - a)}{\rho}\right)\hat{r}_1
- \frac{2\lambda \mu a_0 b^2}{\rho_{2b2}}.
\end{align}
%%%%
Before proceeding, we note the following helpful result.
% \begin{lemma}
% \label{lem:tildeXPlusMinusOne}
% Let $\bar{x}_{1b} = 0 = \bar{x}_{2b}$, $\hat{x}_{1} = 0 = \hat{x}_{2}$, and $\bar{x}_2 = -\bar{x}_1$, with $\hat{\bf b} \neq {\bf 0}$ and $\hat{\bf b}$ not parallel to either $\mathbf{r}_{1b1}$ or $\mathbf{r}_{2b2}$. Then $\bar{x}_1 = 0$ if and only if $\tilde{x} = \pm 1$.
% \end{lemma}
\begin{lemma}
\label{lem:tildeXPlusMinusOne}
Let $\bar{x}_{1b} = 0 = \bar{x}_{2b}$ and $\hat{x}_{1} = 0 = \hat{x}_{2}$, with $\hat{\bf b} \neq {\bf 0}$ and $\hat{\bf b}$ not parallel to either $\mathbf{r}_{1b1}$ or $\mathbf{r}_{2b2}$. Then if $\bar{x}_1 = 0 = \bar{x}_2$, $\tilde{x}$ must either be $1$ or $-1$.
\end{lemma}

\begin{proof}
See the electronic supplementary material.
\end{proof}

As before, we'll first consider a special case.
%%%%
%%%%
\subsection{Case 1 for configuration III: $a_0=0$}
If $a_0=0$, then \eqref{eqn:2Agents2BeaconsDynamicsAtEquilibrium} simplifies to 
\begin{align}
\label{eqn:2Agents2BeaconsDynamics_a0Zero_AtEquilibrium}
\dot{\bar{x}}_1
&= -(1-\lambda)\mu (\bar{x}_1 - a) \left(1 - \bar{x}_1^2 \right)  
+ \frac{1}{\rho} \Bigl(1 - \tilde{x}\Bigr),  \nonumber \\
\dot{\bar{x}}_2
&= -(1-\lambda)\mu (-\bar{x}_1 - a) \left(1 - \bar{x}_1^2 \right)  + \frac{1}{\rho} \Bigl(1 - \tilde{x}\Bigr),  \nonumber \\
\dot{\bar{x}}_{1b1} 
&= -(1-\lambda)\mu (\bar{x}_1 - a)\left(\frac{\rho_{1b1}^2 +\rho^2 - \rho_{2b2}^2  -2\hat{r}_2}{2\rho \rho_{1b1}} \right)  + \frac{1}{\rho_{1b1}} , \nonumber \\
\dot{\bar{x}}_{2b2} 
&= -(1-\lambda)\mu (-\bar{x}_1 - a)\left(\frac{-\rho_{1b1}^2 +\rho^2 + \rho_{2b2}^2  +2\hat{r}_1}{2\rho \rho_{2b2}} \right) + \frac{1}{\rho_{2b2}}, \nonumber \\
\dot{\tilde{x}}
&= - 2\mu(1-\lambda) \bar{x}_1^2\left(1 - \tilde{x}\right),
\nonumber \\
%
%%%%
\dot{\hat{x}}_1 
&= -\mu\left(\frac{(1-\lambda) (\bar{x}_1 - a)}{\rho}\right)\hat{r}_1  
+\left(\frac{(1-\lambda)\mu (\bar{x}_1 - a)}{\rho}\right)\hat{r}_2 
\nonumber \\
%%%%
\dot{\hat{x}}_2 
&= -\mu\left(\frac{(1-\lambda) (-\bar{x}_1 - a)}{\rho}\right)\hat{r}_2  
+\left(\frac{(1-\lambda)\mu (-\bar{x}_1 - a)}{\rho}\right)\hat{r}_1.
\end{align}
As in the same special case for configuration II, the only way that we can have $\dot{\bar{x}}_1$, $\dot{\bar{x}}_2$, and $\dot{\tilde{x}}$ all equal to zero is if $\bar{x}_1=0$. Substituting this constraint into \eqref{eqn:2Agents2BeaconsDynamics_a0Zero_AtEquilibrium} and setting equal to zero yields
\begin{subequations}
\begin{align}
(1-\lambda)\mu a  + \frac{1}{\rho} \Bigl(1 - \tilde{x}\Bigr)
&=
0,
\label{eqn:2Agent2BeaconsDynamicsAtEquilibriumCommonCBParameter_a0Zero__1}
\\
(1-\lambda) \mu a\left(\frac{\rho_{1b1}^2 +\rho^2 - \rho_{2b2}^2 -2\hat{r}_2}{2\rho \rho_{1b1}} \right)
&=
\frac{-1}{\rho_{1b1}}, 
\label{eqn:2Agent2BeaconsDynamicsAtEquilibriumCommonCBParameter_a0Zero__2}
\\
(1-\lambda) \mu a\left(\frac{\rho_{2b2}^2 +\rho^2 - \rho_{1b1}^2 +2\hat{r}_1}{2\rho \rho_{2b2}} \right)
&=
\frac{-1}{\rho_{2b2}}.
\label{eqn:2Agent2BeaconsDynamicsAtEquilibriumCommonCBParameter_a0Zero__3}
\\
\left(\frac{(1-\lambda) \mu a}{\rho}\right)\left(\hat{r}_1 - \hat{r}_2\right) 
&=
0,
\label{eqn:2Agent2BeaconsDynamicsAtEquilibriumCommonCBParameter_a0Zero__4}
\end{align}
\end{subequations}
and since $a=0$ is not valid (based on the fact that $\rho_{1b1}$ and $\rho_{2b2}$ are nonzero and finite), we must have $\hat{r}_1=\hat{r}_2$, $\tilde{x} = -1$ (based on \eqref{eqn:2Agent2BeaconsDynamicsAtEquilibriumCommonCBParameter_a0Zero__1} and \textbf{Lemma \ref{lem:tildeXPlusMinusOne}}), and $\rho = \frac{2}{(1-\lambda)\mu (-a)}$. Then substituting these values into \eqref{eqn:2Agent2BeaconsDynamicsAtEquilibriumCommonCBParameter_a0Zero__2} and \eqref{eqn:2Agent2BeaconsDynamicsAtEquilibriumCommonCBParameter_a0Zero__3} yields the following result.
\begin{proposition}
Consider a beacon-referenced mutual CB pursuit system with beacon positioning parameter $b$ and shape dynamics \eqref{eqn:2AgentShapeDynamicsSet1},\eqref{eqn:2AgentShapeDynamicsSet2} parametrized by $\mu$, $\lambda$, and CB parameters $a$ and $a_0=0$. Then, a \textit{circling equilibrium} exists if and only if $a<0$, and the corresponding equilibrium values satisfy
\begin{equation}
\begin{aligned}
\bullet \quad &
\bar{x}_1 = \bar{x}_2 = 0,
\quad 
\bar{x}_{1b1} = \bar{x}_{2b2} = 0,
\quad 
\hat{x}_1 = \hat{x}_2 = 0,
\quad 
\tilde{x} = -1,
\\
\bullet \quad &
\rho = \frac{2}{(1-\lambda)\mu (-a)},
\quad
\hat{r}_1=\hat{r}_2,
\quad
\rho_{1b1}^2 - \rho_{2b2}^2 -2\hat{r}_1 = 0.
\end{aligned}
\label{eqn:TwoBeaconTwoAgent_a0_zero_CircValues}
\end{equation}
\label{prop:existencePropConfig3Case1}
\end{proposition}
\begin{proof}
Follows from previous discussion.
\end{proof}

%\vspace{.3cm}

We note that if we express $\mathbf{r}_1$ component-wise as $(r_{11}, r_{12}, r_{13})$, then it follows that $r_{13} = \hat{r}_{1}/(2b)$. Therefore the circling equilibria described by \textbf{Proposition \ref{prop:existencePropConfig3Case1}} involve both agents maneuvering in the same plane orthogonal to $\hat{b}$, with circling radius given by $\rho/2 = \frac{1}{(1-\lambda)\mu (-a)}$, and $r_{13} = r_{23}$ determined by initial conditions. Typical circling equilibria corresponding to \textbf{Proposition \ref{prop:existencePropConfig3Case1}} for two different sets of initial conditions are depicted in Figure \ref{fig:TwoBeaconFlatCircle}. 
\begin{figure}
\centering
\begin{subfigure}{.5\textwidth}
  \centering
  \includegraphics[width=.9\linewidth]{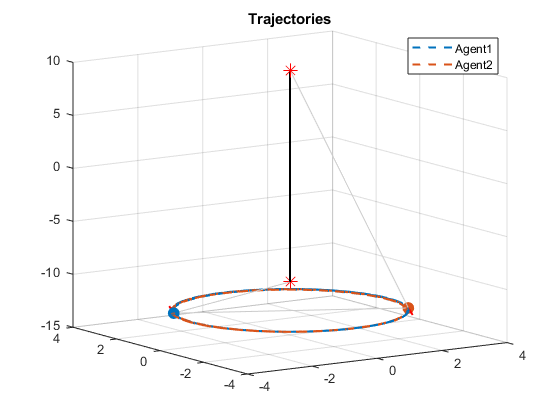}
%   \caption{A subfigure}
  \label{fig:sub1}
\end{subfigure}%
\begin{subfigure}{.5\textwidth}
  \centering
  \includegraphics[width=.9\linewidth]{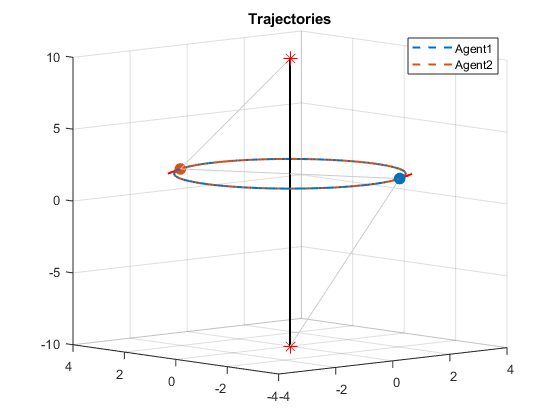}
%   \caption{A subfigure}
  \label{fig:sub2}
\end{subfigure}
\caption{These figures depict circling equilibria for configuration III with beacon positioning $b=10$ and control parameters $a_0=0$, $a = -.771$, $\lambda = 0.57$, and $\mu = 1$, for two different sets of initial conditions.}
\label{fig:TwoBeaconFlatCircle}
\end{figure}
%
%
%
%%%%
%%%%
\subsection{Case 2 for configuration III: Stacked Circling Equilibria}
We now return to the dynamics given by \eqref{eqn:2Agents2BeaconsDynamicsAtEquilibrium} and assume $a_0 \neq 0$. From the form of the $\dot{\tilde{x}}$ expression we note that we can set it equal to zero by choosing $\bar{x}_1=0$. Then \eqref{eqn:2Agents2BeaconsDynamicsAtEquilibrium} simplifies to 
\begin{align}
\label{eqn:2Agents2Beacons_A0_notZero_but_X0_Zero}
\dot{\bar{x}}_1
&= (1-\lambda)\mu a  
+ \lambda\mu a_0\left(\frac{\rho_{1b1}^2 +\rho^2 - \rho_{2b2}^2 -2\hat{r}_2 }{2\rho \rho_{1b1}} \right) + \frac{1}{\rho} \Bigl(1 - \tilde{x}\Bigr),  \nonumber \\
\dot{\bar{x}}_2
&= (1-\lambda)\mu a  + \lambda\mu a_0\left(\frac{-\rho_{1b1}^2 +\rho^2 + \rho_{2b2}^2 +2\hat{r}_1 }{2\rho \rho_{2b2}} \right) + \frac{1}{\rho} \Bigl(1 - \tilde{x}\Bigr),  \nonumber \\
\dot{\bar{x}}_{1b1} 
&= (1-\lambda)\mu a\left(\frac{\rho_{1b1}^2 +\rho^2 - \rho_{2b2}^2  -2\hat{r}_2}{2\rho \rho_{1b1}} \right)  + \frac{1}{\rho_{1b1}} + \lambda\mu a_0, \nonumber \\
\dot{\bar{x}}_{2b2} 
&= (1-\lambda)\mu a\left(\frac{-\rho_{1b1}^2 +\rho^2 + \rho_{2b2}^2  +2\hat{r}_1}{2\rho \rho_{2b2}} \right) + \frac{1}{\rho_{2b2}} + \lambda\mu a_0, \nonumber \\
%%%%
\dot{\hat{x}}_1 
&= \mu\left(\frac{(1-\lambda) a}{\rho} + \frac{\lambda  a_0}{\rho_{1b1}}\right)\hat{r}_1  
-\left(\frac{(1-\lambda)\mu a}{\rho}\right)\hat{r}_2
+ \frac{2\lambda \mu a_0 b^2}{\rho_{1b1}} 
\nonumber \\
%%%%
\dot{\hat{x}}_2 
&= \mu\left(\frac{(1-\lambda) a}{\rho} + \frac{\lambda a_0}{\rho_{2b2}}\right)\hat{r}_2  
-\left(\frac{(1-\lambda)\mu a}{\rho}\right)\hat{r}_1
- \frac{2\lambda \mu a_0 b^2}{\rho_{2b2}}.
\end{align}
%%%%
We note that taking the difference of $\dot{\bar{x}}_1-\dot{\bar{x}}_2$ and the difference of $\dot{\bar{x}}_{1b1} - \dot{\bar{x}}_{2b2}$ yields
\begin{align}
\label{eqn:x1DotDiff}
\dot{\bar{x}}_1-\dot{\bar{x}}_2 
&=
\frac{\lambda\mu a_0}{2\rho}\left(\frac{\rho_{1b1}^2 + \rho^2 - \rho_{2b2}^2 -2\hat{r}_2}{\rho_{1b1}} + \frac{\rho_{1b1}^2 - \rho^2 - \rho_{2b2}^2-2\hat{r}_1}{\rho_{2b2}} \right) 
% \nonumber 
\\
% &=
% \frac{\lambda\mu a_0}{2\rho\rho_{1b1}\rho_{2b2}}\left(\rho_{2b2}\left[\rho_{1b1}^2 + \rho^2 - \rho_{2b2}^2 -2\hat{r}_2\right] + \rho_{1b1}\left[\rho_{1b1}^2 - \rho^2 - \rho_{2b2}^2-2\hat{r}_1\right] \right)
%\nonumber \\
% &=
% \frac{\lambda\mu a_0}{2\rho\rho_{1b1}\rho_{2b2}}\left(\left(\rho_{2b2}-\rho_{1b1}\right)\left[\rho_{1b1}^2 + \rho^2 - \rho_{2b2}^2 -2\hat{r}_2\right] + \rho_{1b1}\left[\rho_{1b1}^2 - \rho^2 - \rho_{2b2}^2-2\hat{r}_1\right] \right) \nonumber \\
\label{eqn:x1bDotDiff}
\dot{\bar{x}}_{1b1} - \dot{\bar{x}}_{2b2}
&=
\frac{(1-\lambda)\mu a}{2\rho}\left(\frac{\rho_{1b1}^2 + \rho^2 - \rho_{2b2}^2 -2\hat{r}_2}{\rho_{1b1}} + \frac{\rho_{1b1}^2 - \rho^2 - \rho_{2b2}^2-2\hat{r}_1}{\rho_{2b2}} \right) 
+ \left(\frac{1}{\rho_{1b1}} - \frac{1}{\rho_{2b2}} \right) 
% \nonumber \\
% &=
% \frac{(1-\lambda)\mu a}{2\rho\rho_{1b1}\rho_{2b2}}\left(\rho_{2b2}\left[\rho_{1b1}^2 + \rho^2 - \rho_{2b2}^2 -2\hat{r}_2\right] + \rho_{1b1}\left[\rho_{1b1}^2 - \rho^2 - \rho_{2b2}^2-2\hat{r}_1\right] \right),
\end{align}
and since $a_0 \neq 0$, setting both equations equal to zero yields the requirement $\rho_{1b1} = \rho_{2b2}$. Substituting this equivalence back into \eqref{eqn:x1DotDiff} and setting equal to zero then requires $\hat{r}_2 = -\hat{r}_1$. Under these constraints,  \eqref{eqn:2Agents2Beacons_A0_notZero_but_X0_Zero} becomes
\begin{align}
\label{eqn:2Agents2Beacons_A0_notZero_but_X0_Zero_V2}
\dot{\bar{x}}_1
&= (1-\lambda)\mu a  
+ \lambda\mu a_0\left(\frac{\rho^2 +2\hat{r}_1 }{2\rho \rho_{1b1}} \right) + \frac{1}{\rho} \Bigl(1 - \tilde{x}\Bigr),  \nonumber \\
\dot{\bar{x}}_{1b1} 
&= (1-\lambda)\mu a\left(\frac{\rho^2 +2\hat{r}_1}{2\rho \rho_{1b1}} \right)  + \frac{1}{\rho_{1b1}} + \lambda\mu a_0, \nonumber \\
%%%%
\dot{\hat{x}}_1 
&= \mu\left(\frac{2(1-\lambda) a}{\rho} + \frac{\lambda  a_0}{\rho_{1b1}}\right)\hat{r}_1  
+ \frac{2\lambda \mu a_0 b^2}{\rho_{1b1}}.
\end{align}
Based on \textbf{Lemma \ref{lem:tildeXPlusMinusOne}}, $\tilde{x}$ must either be $1$ or $-1$. 

%\vspace{.25cm}

\noindent
\textbullet \textbf{ Consider the case where $\mathbf{\tilde{x}=1}$.} Then setting the first equation in \eqref{eqn:2Agents2Beacons_A0_notZero_but_X0_Zero_V2} to zero yields
\begin{align}
     \left(\frac{\rho^2 +2\hat{r}_1}{2\rho\rho_{1b1}}\right)
     = -\left(\frac{1-\lambda}{\lambda}\right) \left(\frac{a}{a_0}\right),
\end{align}
from which substitution into the second equation in  \eqref{eqn:2Agents2Beacons_A0_notZero_but_X0_Zero_V2} and setting equal to zero yields
\begin{align}
    \frac{1}{\rho_{1b1}}
    &= \mu\left(\frac{(1-\lambda)^2 a^2}{\lambda a_0} \right) -\lambda \mu a_0 = \mu \left(\frac{(1-\lambda)^2 a^2 - \lambda^2 a_0^2}{\lambda a_0} \right),
\end{align}
i.e.
\begin{align}
    \rho_{1b1} = \frac{\lambda a_0}{\mu\Bigl((1-\lambda)^2 a^2 - \lambda^2 a_0^2\Bigr)}.
\end{align}
Note that this requires $\frac{a_0}{(1-\lambda)^2 a^2 - \lambda^2 a_0^2} > 0$. Denoting 
\begin{equation}
    \Phi =  \frac{(1-\lambda)a}{\mu\Bigl((1-\lambda)^2 a^2 - \lambda^2 a_0^2\Bigr)}
\end{equation}
so that $\rho_{1b1} = \frac{\lambda a_0}{(1-\lambda)a}\Phi$, we can substitute back into the first and third equations in \eqref{eqn:2Agents2Beacons_A0_notZero_but_X0_Zero_V2} and set them to zero to obtain
\begin{align}
    0 &= 2\rho\Phi + \rho^2 + 2\hat{r}_1, \nonumber \\
    0 &= 2\hat{r}_1 \Phi + \rho(\hat{r}_1 + 2 b^2).
\end{align}
From the first equation we have $\hat{r}_1 = -\rho\left(\frac{\rho}{2}+\Phi \right)$, and substitution into the second equation results in 
\begin{align}
    0 &= -2\rho\left(\frac{\rho}{2}+\Phi \right) \Phi + \rho\left(-\rho\left(\frac{\rho}{2}+\Phi \right) + 2 b^2\right)
    = -\frac{\rho}{2}\Bigl(\rho^2 + 4\Phi \rho + 4\Phi^2 - 4b^2 \Bigr).
\end{align}
Solving the quadratic equation inside the parentheses then results in 
\begin{equation}
\label{eqn:rhoEquilibriumValue2Agents2Beacons_A0_notZero_but_X0_Zero}
\rho = 2\left(-\Phi \pm b \right),
\end{equation}
which implies that $\frac{\rho}{2}+\Phi = \pm b$, and therefore we also have
\begin{equation}
\label{eqn:r1hatEquilibriumValue2Agents2Beacons_A0_notZero_but_X0_Zero}    
\hat{r}_1 = \mp b \rho = 2b\left(-b \pm \Phi \right).
\end{equation}
(Note that the choice of $+$ or $-$ in the $\hat{r}_1$ equation must match the choice from the $\rho$ equation, i.e. we have only two options rather than four combinations.) To characterize existence conditions, we must ensure that \eqref{eqn:rhoEquilibriumValue2Agents2Beacons_A0_notZero_but_X0_Zero} satisfies $\rho > 0$ and that the constraint \eqref{eqn:rUnitb2Unit} (equivalent to \eqref{eqn:rUnitb1Unit} in the present case) is satisfied. Substituting $\rho_{1b1} = \rho_{2b2} = \frac{\lambda a_0}{(1-\lambda)a}\Phi$ into \eqref{eqn:rUnitb2Unit} results in the requirement
\begin{align}
    \frac{-(1-\lambda)a}{2\rho\lambda a_0 \Phi}\left(\rho^2  +2\hat{r}_1\right) \in [-1,1],
\end{align}
and since $\rho^2  +2\hat{r}_1 = -2\rho\Phi$, we obtain the parameter constraint
\begin{equation}
    \left|\frac{(1-\lambda)a}{\lambda a_0} \right| < 1.
\end{equation}
Note that this implies that the denominator of $\Phi$ is negative, and therefore $\sgn(\Phi) = -\sgn(a)$. Since the denominator of $\rho_{1b1}$ is the same as the denominator of $\Phi$ and we require $\rho_{1b1} = \rho_{2b2} > 0$, we must have $a_0 < 0$. Then in order for $\rho>0$, then we either must have $\Phi < 0$ (i.e. $\sgn(a)>0$) or $\Phi >0$ (i.e. $\sgn(a)<0$) with $b>\Phi$. Representative trajectories for this case are depicted in Figure \ref{fig:TwoBeaconStackedCircle}.

\begin{figure}
\centering
\begin{subfigure}{.5\textwidth}
  \centering
  \includegraphics[width=.9\linewidth]{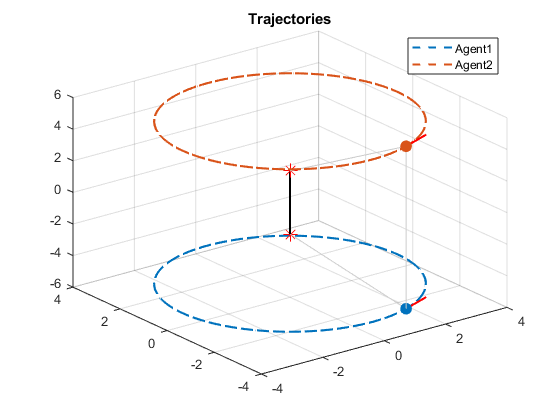}
%   \caption{A subfigure}
  \label{fig:sub1}
\end{subfigure}%
\begin{subfigure}{.5\textwidth}
  \centering
  \includegraphics[width=.9\linewidth]{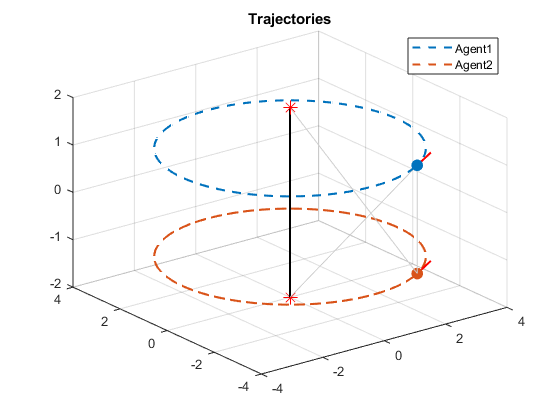}
%   \caption{A subfigure}
  \label{fig:sub2}
\end{subfigure}
\caption{These figures depict circling equilibria for configuration III with beacon positioning $b=2$ and control parameters $a_0=-.707$, $a = .707$, $\lambda = 0.6$, and $\mu = 0.9$, with two different sets of initial conditions. Note that since $\Phi \approx -5.66$ is negative, we have two possible steady-state configurations for this set of parameters. If $\Phi>0$, then there would only be one possible steady-state configuration.}
\label{fig:TwoBeaconStackedCircle}
\end{figure}

We now return to our earlier choice for $\tilde{x}$ and consider the other option.

%\vspace{.25cm}

\noindent
\textbullet \textbf{ Consider the case where $\mathbf{\tilde{x}=-1}$.} As we have done in a previous context, we will denote
\begin{equation}
    \tilde{a}_{0} \triangleq \lambda\mu a_0 \in \mathds{R}, \qquad \tilde{a} \triangleq (1-\lambda)\mu a \in \mathds{R}.
\end{equation}
Then substituting $\tilde{x}=-1$ into \eqref{eqn:2Agents2Beacons_A0_notZero_but_X0_Zero_V2} and setting all three equations to zero (and eliminating denominators) yields
\begin{align}
0
&= 2 \tilde{a}\rho\rho_{1b1} + \tilde{a}_0 \left(\rho^2 +2\hat{r}_1\right) + 4\rho_{1b1},   
\label{eqn:2Agents2Beacons_A0_notZero_but_X0_Zero_V3_eqn1}
\\
0
&= 2 \tilde{a}_{0}\rho\rho_{1b1} + \tilde{a}\left(\rho^2 +2\hat{r}_1\right) + 2\rho,  
\label{eqn:2Agents2Beacons_A0_notZero_but_X0_Zero_V3_eqn2}
\\
0
&= 2 \tilde{a}\rho_{1b1}\hat{r}_1 + \tilde{a}_{0} \rho \hat{r}_1 + 2\tilde{a}_0 \rho b^2 .
\label{eqn:2Agents2Beacons_A0_notZero_but_X0_Zero_V3_eqn3}
\end{align}
As in the previous case, we have $\rho_{1b1} = \rho_{2b2}$ and $\hat{r}_2 = -\hat{r}_1$. Taking the sum and difference of \eqref{eqn:2Agents2Beacons_A0_notZero_but_X0_Zero_V3_eqn1} and \eqref{eqn:2Agents2Beacons_A0_notZero_but_X0_Zero_V3_eqn2} results in
\begin{align}
0
&= \left(\tilde{a}_0 + \tilde{a} \right) \left(2\rho\rho_{1b1} + \rho^2 +2\hat{r}_1\right) + 2\left(2\rho_{1b1} + \rho\right),   
\label{eqn:2Agents2Beacons_A0_notZero_but_X0_Zero_V4_eqn1}
\\
0
&= \left(\tilde{a}_0 - \tilde{a} \right) \left(-2\rho\rho_{1b1} + \rho^2 +2\hat{r}_1\right) + 2\left(2\rho_{1b1} - \rho\right),  
\label{eqn:2Agents2Beacons_A0_notZero_but_X0_Zero_V4_eqn2}
\end{align}
which suggests the change of variables $\rho_{_+} = \frac{\rho}{2} + \rho_{1b1}$ and $\rho_{_-} = \frac{\rho}{2} - \rho_{1b1}$, i.e. $\rho = \rho_{_+} + \rho_{_-}$ and $\rho_{1b1} = (\rho_{_+} - \rho_{_-})/2$. In terms of this notation, we have $2\rho\rho_{1b1} = \rho_{_+}^2 - \rho{-}^2$ and $\rho^2 = (\rho_{_+} + \rho_{_-})^2$, and therefore we can express \eqref{eqn:2Agents2Beacons_A0_notZero_but_X0_Zero_V4_eqn1}, \eqref{eqn:2Agents2Beacons_A0_notZero_but_X0_Zero_V4_eqn2}, and \eqref{eqn:2Agents2Beacons_A0_notZero_but_X0_Zero_V3_eqn3} as
\begin{align}
0
&= \left(\tilde{a}_0 + \tilde{a} \right) \left(\rho_{_+}^2 + \rho_{_+}\rho_{_-}  +\hat{r}_1\right) + 2\rho_{_+},   
\label{eqn:2Agents2Beacons_A0_notZero_but_X0_Zero_V5_eqn1}
\\
0
&=  \left(\tilde{a}_0 - \tilde{a} \right) \left(\rho_{_-}^2 + \rho_{_+}\rho_{_-}  +\hat{r}_1\right) - 2\rho_{_-},  
\label{eqn:2Agents2Beacons_A0_notZero_but_X0_Zero_V5_eqn2}
\\
0
&= (\tilde{a}_0 + \tilde{a})\hat{r}_1\rho_{_+} + (\tilde{a}_0 - \tilde{a})\hat{r}_1\rho_{_-} + 2\tilde{a}_0 b^2 (\rho_{_+} + \rho_{_-})  .
\label{eqn:2Agents2Beacons_A0_notZero_but_X0_Zero_V5_eqn3}
\end{align}
Before proceeding, we note that $\rho_{_+}$ must be positive (since $\rho>0$ and $\rho_{1b1}>0$), and therefore \eqref{eqn:2Agents2Beacons_A0_notZero_but_X0_Zero_V5_eqn1} implies that $\rho_{_+}^2 + \rho_{_+}\rho_{_-}  +\hat{r}_1$ must be nonzero. Additionally, if $\rho_{_-}^2 + \rho_{_+}\rho_{_-}  +\hat{r}_1=0$, then \eqref{eqn:2Agents2Beacons_A0_notZero_but_X0_Zero_V5_eqn2} requires $\rho_{_-} = 0$, and the combination of the two constraints results in $\hat{r}_1 = 0$. But substituting both of these values into \eqref{eqn:2Agents2Beacons_A0_notZero_but_X0_Zero_V5_eqn3} results in $2\tilde{a}_0 b^2 \rho_{_+} = 0$, which is not possible since $\tilde{a}_0$, $b$, and $\rho_{_+}$ are all nonzero by assumption. Therefore it is valid to rearrange \eqref{eqn:2Agents2Beacons_A0_notZero_but_X0_Zero_V5_eqn1} and \eqref{eqn:2Agents2Beacons_A0_notZero_but_X0_Zero_V5_eqn2} to obtain
\begin{align}
\label{eqn:aTildeSumAndDiffEqns}
    \tilde{a}_0 + \tilde{a} &= \frac{-2\rho_{_+}}{\rho_{_+}^2 + \rho_{_+}\rho_{_-}  +\hat{r}_1}, \nonumber \\
    \tilde{a}_0 - \tilde{a} &= \frac{2\rho_{_-}}{\rho_{_-}^2 + \rho_{_+}\rho_{_-}  +\hat{r}_1}.
\end{align}
Noting that summing the two equations in \eqref{eqn:aTildeSumAndDiffEqns} yields $2\tilde{a}_0$, we can substitute these expressions into \eqref{eqn:2Agents2Beacons_A0_notZero_but_X0_Zero_V5_eqn3} to obtain
\begin{align}
 0
%  &= \left(\frac{-2\rho_{_+}}{\rho_{_+}^2 + \rho_{_+}\rho_{_-}  +\hat{r}_1}\right)\hat{r}_1\rho_{_+} + \left(\frac{2\rho_{_-}}{\rho_{_-}^2 + \rho_{_+}\rho_{_-}  +\hat{r}_1}\right)\hat{r}_1\rho_{_-} + 2\tilde{a}_0 b^2 (\rho_{_+} + \rho_{_-}) \nonumber \\
&= 2\hat{r}_1\left(\frac{-\rho_{_+}^2}{\rho_{_+}^2 + \rho_{_+}\rho_{_-}  +\hat{r}_1} + \frac{\rho_{_-}^2}{\rho_{_-}^2 + \rho_{_+}\rho_{_-}  +\hat{r}_1}\right)
+ 2b^2 (\rho_{_+} + \rho_{_-})\left(\frac{\rho_{_-}}{\rho_{_-}^2 + \rho_{_+}\rho_{_-}  +\hat{r}_1} - \frac{\rho_{_+}}{\rho_{_+}^2 + \rho_{_+}\rho_{_-}  +\hat{r}_1}\right),
\end{align}
which implies that
\begin{align}
0 &= \hat{r}_1\biggl(-\rho_{_+}^2 \left(\rho_{_-}^2 + \rho_{_+}\rho_{_-}  +\hat{r}_1\right) + \rho_{_-}^2 \left(\rho_{_+}^2 + \rho_{_+}\rho_{_-}  +\hat{r}_1\right)\biggr) 
  \nonumber \\
  &\qquad
  + b^2 (\rho_{_+} + \rho_{_-})\biggl(\rho_{_-} \left(\rho_{_+}^2 + \rho_{_+}\rho_{_-}  +\hat{r}_1\right) - \rho_{_+} \left(\rho_{_-}^2 + \rho_{_+}\rho_{_-}  +\hat{r}_1\right)\biggr)  \nonumber \\
 &= \hat{r}_1 \left(\rho_{_+}\rho_{_-}  +\hat{r}_1\right) \left(\rho_{_-}^2 - \rho_{_+}^2 \right) 
  + b^2 (\rho_{_+} + \rho_{_-}) (\rho_{_-} - \rho_{_+})\hat{r}_1  \nonumber \\
 &= \hat{r}_1 \left(\rho_{_-}^2 - \rho_{_+}^2 \right)\left(\rho_{_+}\rho_{_-}  +\hat{r}_1  + b^2 \right).
\end{align}
Since $\rho_{_-}^2 - \rho_{_+}^2 = -2\rho\rho_{1b1} \neq 0$ and it is straightforward to show that $\hat{r}_1=0$ leads to a contradiction, we must have $\hat{r}_1 = -\rho_{_+}\rho_{_-}  -b^2$. Substitution into \eqref{eqn:2Agents2Beacons_A0_notZero_but_X0_Zero_V5_eqn1} and \eqref{eqn:2Agents2Beacons_A0_notZero_but_X0_Zero_V5_eqn2} then results in
\begin{align}
0
&= \left(\tilde{a}_0 + \tilde{a} \right) \left(\rho_{_+}^2 -b^2\right) + 2\rho_{_+},   
\label{eqn:2Agents2Beacons_A0_notZero_but_X0_Zero_V6_eqn1}
\\
0
&=  \left(\tilde{a}_0 - \tilde{a} \right) \left(\rho_{_-}^2 -b^2\right) - 2\rho_{_-}.  
\label{eqn:2Agents2Beacons_A0_notZero_but_X0_Zero_V6_eqn2}
\end{align}
If we represent our constraint \eqref{eqn:rUnitb2Unit} in terms of the $\rho_{_+}$ and $\rho_{_-}$ variables, it is straightforward to obtain the requirement $\rho_{_+}^2 > b^2 > \rho_{_-}^2$. Thus the form of \eqref{eqn:2Agents2Beacons_A0_notZero_but_X0_Zero_V6_eqn1} imposes the requirement $\tilde{a}_0 + \tilde{a} < 0$, and solving the quadratic equation in $\rho_{_+}$ (and selecting the only option that corresponds to $\rho_{_+}>0$) leads to
\begin{equation}
    \rho_{_+} = \frac{1}{-(\tilde{a}_0 + \tilde{a})} \left(1 + \sqrt{1 + \Bigl((\tilde{a}_0 + \tilde{a})b\Bigr)^2} \right).
\end{equation}
If $\tilde{a}_0 - \tilde{a} = 0$, then \eqref{eqn:2Agents2Beacons_A0_notZero_but_X0_Zero_V6_eqn2} requires $\rho_{_-}=0$. If $\tilde{a}_0 - \tilde{a} \neq 0$, then \eqref{eqn:2Agents2Beacons_A0_notZero_but_X0_Zero_V6_eqn2} can be arranged in the form
\begin{equation}
    \rho_{_-}^2 -\left(\frac{2}{\tilde{a}_0 - \tilde{a}}\right) \rho_{_-} - b^2 = 0,
\end{equation}
with constraint \eqref{eqn:rUnitb2Unit} requiring $(\tilde{a}_0 - \tilde{a})\rho_{_-}  < 0$ to ensure $b^2 > \rho_{_-}^2$. This leads to the result
\begin{equation}
    \rho_{_-} = \frac{1}{\tilde{a}_0 - \tilde{a}} \left(1 - \sqrt{1 + \Bigl((\tilde{a}_0 - \tilde{a})b\Bigr)^2} \right).
\end{equation}
The following proposition summarizes the results of this entire subsection. 

\begin{proposition}
Consider a beacon-referenced mutual CB pursuit system with beacon positioning parameter $b$ and shape dynamics \eqref{eqn:2AgentShapeDynamicsSet1}-\eqref{eqn:2AgentShapeDynamicsSet2} parametrized by $\mu$, $\lambda$, and CB parameters $a$ and $a_0 \neq 0$, and define
\begin{equation}
    \Phi =  \frac{(1-\lambda)a}{\mu\Bigl((1-\lambda)^2 a^2 - \lambda^2 a_0^2\Bigr)}
\end{equation}
Circling equilibria exist under the following conditions, with equilibrium values in each case given by
\begin{equation}
\begin{aligned}
\bullet \quad &
\bar{x}_1 = \bar{x}_2 = 0,
\quad 
\bar{x}_{1b} = \bar{x}_{2b} = 0.
\end{aligned}
\end{equation}
\indent (a) If $a_0 <0$, $a<0$, $(1-\lambda)^2 a^2 - \lambda^2 a_0^2 < 0$, and $\Phi < b$, a circling equilibrium exists with corresponding equilibrium values given by
\begin{equation}
\begin{aligned}
\bullet \quad &
\tilde{x} = 1,
\quad 
\hat{r}_1 = -\hat{r}_2 = 2b\left(-b + \Phi \right),
\\
\bullet \quad &
\rho = 2\left(-\Phi + b \right),
% \\
% \bullet \quad &
\quad
\rho_{1b1} = \rho_{2b2} 
= \left(\frac{\lambda a_0}{(1-\lambda)a}\right)\Phi.
\end{aligned}
\end{equation}
\\
\indent (b) If $a_0 <0$, $a>0$, and $(1-\lambda)^2 a^2 - \lambda^2 a_0^2$, a pair of circling equilibria exist with corresponding equilibrium values given by
\begin{equation}
\begin{aligned}
\bullet \quad &
\tilde{x} = 1,
\quad 
\hat{r}_1 = -\hat{r}_2 = 2b\left(-b \pm \Phi \right),
\\
\bullet \quad &
\rho = 2\left(-\Phi \pm b \right),
% \\
% \bullet \quad &
\quad
\rho_{1b1} = \rho_{2b2} 
= \left(\frac{\lambda a_0}{(1-\lambda)a}\right)\Phi.
\end{aligned}
\end{equation}
\\
\indent (c) If $a_0 <0$, and $(1-\lambda) a = \lambda a_0$, a circling equilibrium exists with corresponding equilibrium values given by
\begin{equation}
\begin{aligned}
\bullet \quad &
\tilde{x} = -1,
\quad 
\hat{r}_1 = -\hat{r}_2 = -b^2,
\\
\bullet \quad &
\rho =  \frac{1}{-(2\lambda \mu a_0)} \left(1 + \sqrt{1 + (2\lambda \mu a_0b)^2} \right),
% \\
% \bullet \quad &
\quad
\rho_{1b1} = \rho_{2b2} 
= \rho/2.
\end{aligned}
\end{equation}
\\
\indent (d) If $(1-\lambda) a + \lambda a_0 <0$ and $(1-\lambda) a - \lambda a_0 \neq 0$, a circling equilibrium exists with corresponding equilibrium values given by
\begin{equation}
\begin{aligned}
\bullet \quad &
\tilde{x} = -1,
\quad 
\hat{r}_1 = -\hat{r}_2 = -\rho_{_+}\rho_{_-} -b^2,
\\
\bullet \quad &
\rho_{_+} = \frac{1}{-\mu\bigl((1-\lambda) a + \lambda a_0\bigr)} \left(1 + \sqrt{1 + \Bigl(\mu\bigl((1-\lambda) a + \lambda a_0\bigr)b\Bigr)^2} \right),
\\
\bullet \quad &
\rho_{_-} = \frac{1}{\mu\bigl((1-\lambda) a - \lambda a_0\bigr)} \left(1 - \sqrt{1 + \Bigl(\mu\bigl((1-\lambda) a - \lambda a_0\bigr)b\Bigr)^2} \right),
\\
\bullet \quad &
\rho = \rho_{_+} + \rho_{_-},
\quad
\rho_{1b1} = \rho_{2b2} 
= (\rho_{_+} - \rho_{_-})/2.
\end{aligned}
\end{equation}
\label{prop:existencePropConfig3Case2}
\end{proposition}
\begin{proof}
Follows from previous discussion.
\end{proof}
%

%\vspace{.3cm}

\noindent
{\bf Remark:} In each of the cases above, the third component of $r_1$ is given by $\hat{r}_1/(2b)$, and therefore the $\hat{r}_1$ values indicate the ``vertical'' displacement of the planes of orbit. Radii for the circling orbits can be determined by projecting the vector $\mathbf{r}_{1b1}$ onto the ``x-y'' plane (i.e. the plane normal to $\hat{\bf b}$ which passes through the origin) which yields the expression
\begin{equation}
    radius = \sqrt{\rho_{1b1}^2 - \left(\frac{\hat{r}_1 + 2 b^2}{2b} \right)^2}.
\end{equation}
Representative trajectories for the equilibria described in the last two bullets of \textbf{Proposition \ref{prop:existencePropConfig3Case2}} are depicted in Figure \ref{fig:TwoBeaconCrossStackCircle}.

\begin{figure}
\centering
\begin{subfigure}{.5\textwidth}
  \centering
  \includegraphics[width=.9\linewidth]{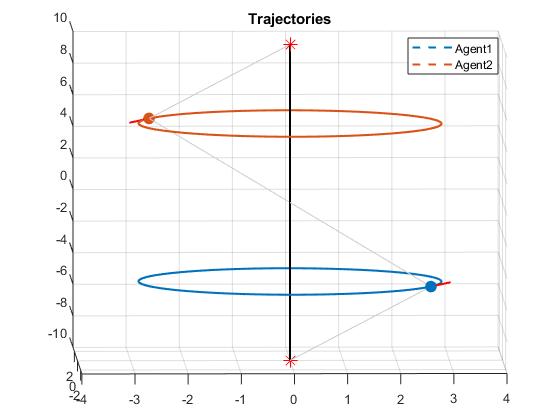}
%   \caption{A subfigure}
  \label{fig:sub1A}
\end{subfigure}%
\begin{subfigure}{.5\textwidth}
  \centering
  \includegraphics[width=.9\linewidth]{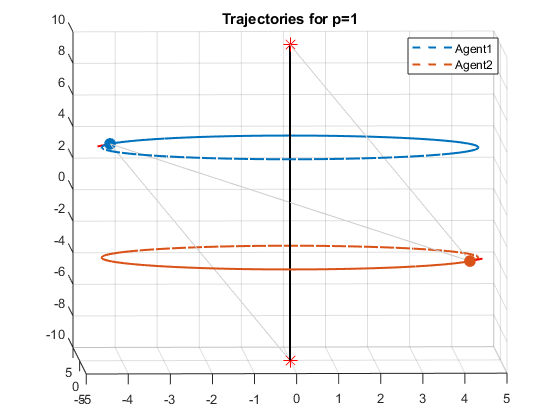}
%   \caption{A subfigure}
  \label{fig:sub2A}
\end{subfigure}
\caption{These figures depict circling equilibria corresponding to the last two bullets of \textbf{Proposition \ref{prop:existencePropConfig3Case2}} with beacon positioning $b=10$. The trajectories in the figure on the left were generated using control parameters $a_0=-0.707$, $a = -0.707$, $\lambda = 0.5$, and $\mu = 1$, and correspond to bullet (c) of \textbf{Proposition \ref{prop:existencePropConfig3Case2}}. The trajectories in the figure on the right were generated using control parameters $a_0=0.707$, $a = -0.588$, $\lambda = 0.35$, and $\mu = 1$, and correspond to bullet (d) of the proposition.}
\label{fig:TwoBeaconCrossStackCircle}
\end{figure}
\noindent
{\bf Remark:} It is important to note that \textbf{Proposition \ref{prop:existencePropConfig3Case2}} does not provide an exhaustive characterization of circling equilibria for configuration III, in that it provides sufficient (but not necessary) conditions for existence. Numerical simulations suggest an additional type of circling equilibria for this configuration, for which $\bar{x}_1 \neq 0$ and the midpoint between the agents moves on a separate circling trajectory around the beacon axis. Characterization of this type of circling equilibria will be the focus of future work.

\section{Conclusion}
In this work we have proposed a new control law which references multiple targets and relies only on bearing measurements. Closed-loop shape dynamics were formulated for the multiple configuration possibilities for 1-2 agents referencing 1-2 fixed beacons, and these shape dynamics were analyzed to determine existence conditions and steady-state characterization for circling relative equilibria. In each case, we demonstrated that the radius of the circling trajectory and the vertical separation of the mobile agents could be prescribed through the choice of control parameters. This decentralized control method could be used to coordinate the motions of autonomous vehicles in a circular stationing pattern with minimal required sensing (e.g. underwater vehicles sensing the relative bearing to sound sources serving as beacons).  

There exists several clear paths for future research endeavors related to this control formulation. First, it will be important to explore stability characteristics of these steady-state behaviors to determine parameter requirements to ensure that the circling equilibria are attractive. Numerical simulations suggest that for most of the circling equilibria described in this paper there exists a range of parameter values for which the circling trajectories are attractive, and future work will prove this through mathematical analysis. Additionally, the ideas in this paper can be extended in new directions by considering systems with more than two mobile agents, each referencing two (or more) targets.

{\bf Funding:} K. S. Galloway was supported by funding from Naval Research Laboratory and from the U.S. Naval Academy.

{\bf Acknowledgement:} The authors wish to thank Prof. Levi DeVries for helpful discussions and feedback on this paper.
\bibliography{RefsForPRS2019}
\bibliographystyle{iclr2020_conference}

\appendix
\appendixpage

\section{Supplemental Calculations}

\counterwithout{equation}{section}
\setcounter{equation}{0}
\renewcommand\theequation{A.\arabic{equation}}

These calculations are provided for a more detailed proof of several of the claims and propositions. 
%We adopt the following equation numbering convention: equation numbers which include a section number refer to the main paper, but those without a section number appear only in these Supplemental Calculations.
%
%%%%%%%%%%%%%%%%%%%%%%%%%%%%%%%%%%
\subsection{Derivation of shape dynamics \eqref{eqn:shapeDynamicsForOneAgentTwoBeacons} for Configuration I}
First, we can calculate 
\begin{equation}
    \dot{\rho}_{1bi} = \frac{d}{dt}\left(\left|{\bf r}_{1bi}\right| \right)
    = \dot{\bf r}_{1bi} \cdot \rOneBiUnit
    = \mathbf{x}_1 \cdot \rOneBiUnit
    = \bar{x}_{1bi}, \; i=1,2.
\end{equation}
Next, it can be shown that  
\begin{equation}
    \frac{d}{dt}\left(\rOneBiUnit \right) = \frac{1}{\left|{\bf r}_{1bi}\right|}\left(\mathbf{x}_1 - \bar{x}_{1bi} \rOneBiUnit\right), \; i=1,2,
\end{equation}
and therefore we have 
\begin{align}
    \dot{\bar{x}}_{1b1} &= \dot{\bf x}_1 \cdot \rOneBUnit + \mathbf{x}_1 \cdot \frac{d}{dt}\left(\rOneBUnit \right) \nonumber \\
    &=
    \left[-(1-\lambda)\mu (\bar{x}_{1b2} - a_{b2})\left(\rOneBTwoUnit - \bar{x}_{1b2}{\bf x}_1\right)- \lambda \mu (\bar{x}_{1b1} - a_{b1})\left(\rOneBUnit - \bar{x}_{1b1}{\bf x}_1\right) \right] \cdot \rOneBUnit \nonumber \\
    &\qquad + \mathbf{x}_1 \cdot \left[\frac{1}{\left|{\bf r}_{1b1}\right|}\left(\mathbf{x}_1 - \bar{x}_{1b1} \rOneBUnit\right) \right] \nonumber \\
    &=
    -(1-\lambda)\mu (\bar{x}_{1b2} - a_{b2})\left(\rOneBTwoUnit\cdot \rOneBUnit - \bar{x}_{1b1}\bar{x}_{1b2}\right)- \lambda \mu (\bar{x}_{1b1} - a_{b1})\left(1 - \bar{x}_{1b1}^2\right) + \frac{1}{\rho_{1b1}}\left(1 - \bar{x}_{1b1}^2 \right) \nonumber \\
    &= -(1-\lambda)(\mu (\bar{x}_{1b2} - a_{b2}))\left(\rOneBUnit \cdot \rOneBTwoUnit - \bar{x}_{1b1}\bar{x}_{1b2} \right)  -\left(\lambda\mu(\bar{x}_{1b1}-a_{b1}) - \frac{1}{\rho_{1b1}}\right)\Bigl(1 - \bar{x}_{1b1}^2 \Bigr),
\end{align}
and
\begin{align}
    \dot{\bar{x}}_{1b2} &= \dot{\bf x}_1 \cdot \rOneBTwoUnit + \mathbf{x}_1 \cdot \frac{d}{dt}\left(\rOneBTwoUnit \right) \nonumber \\
    &=
    \left[-(1-\lambda)\mu (\bar{x}_{1b2} - a_{b2})\left(\rOneBTwoUnit - \bar{x}_{1b2}{\bf x}_1\right)- \lambda \mu (\bar{x}_{1b1} - a_{b1})\left(\rOneBUnit - \bar{x}_{1b1}{\bf x}_1\right) \right] \cdot \rOneBTwoUnit \nonumber \\
    &\qquad + \mathbf{x}_1 \cdot \left[\frac{1}{\left|{\bf r}_{1b2}\right|}\left(\mathbf{x}_1 - \bar{x}_{1b2} \rOneBTwoUnit\right) \right] \nonumber \\
    &=
    -(1-\lambda)\mu (\bar{x}_{1b2} - a_{b2})\left(1 - \bar{x}_{1b1}^2\right)- \lambda \mu (\bar{x}_{1b1} - a_{b1})\left(\rOneBTwoUnit\cdot \rOneBUnit - \bar{x}_{1b1}\bar{x}_{1b2}\right) + \frac{1}{\rho_{1b2}}\left(1 - \bar{x}_{1b2}^2 \right) \nonumber \\
    &= -\lambda\mu (\bar{x}_{1b1} - a_{b1}))\left(\rOneBUnit \cdot \rOneBTwoUnit - \bar{x}_{1b1}\bar{x}_{1b2} \right)  -\left((1-\lambda)\mu(\bar{x}_{1b2}-a_{b2}) - \frac{1}{\rho_{1b2}}\right)\Bigl(1 - \bar{x}_{1b2}^2 \Bigr).
\end{align}
%
%%%%%%%%%%%%%%%%%%%%%%%%%%%%%%%%%%
\subsection{Derivation of shape dynamics \eqref{eqn:2AgentShapeDynamicsSet1} for Configuration II and III}
The derivation of the dynamics for $\rho_{1b1}$, $\rho_{2b2}$, $\bar{x}_{1b1}$, and $\bar{x}_{2b2}$ are very similar to the derivation presented for \eqref{eqn:shapeDynamicsForOneAgentTwoBeacons}, and therefore we won't repeat them here. By direct calculation, we have
\begin{equation}
    \dot{\rho} = \frac{d}{dt}\left(\left|{\bf r}\right| \right)
    = \dot{\bf r} \cdot \rUnit
    = (\mathbf{x}_1 - \mathbf{x}_2) \cdot \rUnit
    = \bar{x}_{1} + \bar{x}_{2}.
\end{equation}
Then, noting that
\begin{equation}
    \frac{d}{dt}\left(\rUnit \right) = \frac{1}{\left|{\bf r}\right|}\left[\dot{\bf r} -  \left(\dot{\bf r} \cdot \rUnit \right)\rUnit\right]
    =
    \frac{1}{\rho}\left({\bf x}_1 - {\bf x}_2 -  \left(\bar{x}_1 + \bar{x}_2 \right)\rUnit\right), 
\end{equation}
starting from \eqref{eqn:closedLoopDynamics}, we have
\begin{align}
    \dot{\bar{x}}_{1} &= 
    \dot{\bf x}_1 \cdot \rUnit + \mathbf{x}_1 \cdot \frac{d}{dt}\left(\rUnit \right) \nonumber \\
    &=
    \left[-(1-\lambda)\mu (\bar{x}_1 - a)\left(\rUnit - \bar{x}_1{\bf x}_1\right)   
- \lambda \mu (\bar{x}_{1b1} - a_0)\left(\rOneBUnit - \bar{x}_{1b1} {\bf x}_1\right)\right] \cdot \rUnit \nonumber \\
&\qquad + \mathbf{x}_1 \cdot \left[\frac{1}{\rho}\left({\bf x}_1 - {\bf x}_2 -  \left(\bar{x}_1 + \bar{x}_2 \right)\rUnit\right)\right] 
    \nonumber \\
    &= -(1-\lambda)\mu (\bar{x}_1 - a) \left(1 - \bar{x}_1^2 \right)  -\lambda\mu(\bar{x}_{1b1}-a_0)\left(\rOneBUnit \cdot \rUnit - \bar{x}_{1b1}\bar{x}_1 \right) + \frac{1}{\rho} \Bigl(1 - \tilde{x} - \bar{x}_1^2 - \bar{x}_1 \bar{x}_{2} \Bigr),
\end{align}
and
\begin{align}
\dot{\bar{x}}_{2} &= 
    \dot{\bf x}_2 \cdot \left(-\rUnit\right) + \mathbf{x}_2 \cdot \frac{d}{dt}\left(-\rUnit \right) \nonumber \\
    &=
    \left[-(1-\lambda)\mu (\bar{x}_2 - a)\left(-\rUnit - \bar{x}_2{\bf x}_2\right) 
- \lambda \mu (\bar{x}_{2b2} - a_0)\left(\rTwoBUnit - \bar{x}_{2b2} {\bf x}_2\right)\right] \cdot \left(-\rUnit\right) \nonumber \\
&\qquad - \mathbf{x}_2 \cdot \left[\frac{1}{\rho}\left({\bf x}_1 - {\bf x}_2 -  \left(\bar{x}_1 + \bar{x}_2 \right)\rUnit\right)\right] 
    \nonumber \\
    &= -(1-\lambda)\mu (\bar{x}_2 - a) \left(1 - \bar{x}_2^2 \right)  -\lambda\mu(\bar{x}_{2b2}-a_0)\left(-\rTwoBUnit \cdot \rUnit - \bar{x}_{2b2}\bar{x}_2 \right) + \frac{1}{\rho} \Bigl(1 - \tilde{x} - \bar{x}_2^2 - \bar{x}_1 \bar{x}_{2} \Bigr).
\end{align}

Lastly, we can derive the dynamics for $\tilde{x}$ by calculating
\begin{align}
    \dot{\tilde{x}} 
    &= \dot{\bf x}_{1} \cdot {\bf x}_{2} + {\bf x}_{1} \cdot \dot{\bf x}_{2} \nonumber \\
    &= \left[-(1-\lambda)\mu (\bar{x}_1 - a)\left(\rUnit - \bar{x}_1{\bf x}_1\right)   
- \lambda \mu (\bar{x}_{1b1} - a_0)\left(\rOneBUnit - \bar{x}_{1b1} {\bf x}_1\right)\right] \cdot {\bf x}_{2} \nonumber \\
&\qquad + {\bf x}_{1} \cdot \left[-(1-\lambda)\mu (\bar{x}_2 - a)\left(-\rUnit - \bar{x}_2{\bf x}_2\right) 
- \lambda \mu (\bar{x}_{2b2} - a_0)\left(\rTwoBUnit - \bar{x}_{2b2} {\bf x}_2\right)\right] \nonumber \\
&= (1-\lambda)\mu (\bar{x}_1 - a) \left(\bar{x}_{2} + \tilde{x} \bar{x}_1\right)  + (1-\lambda)\mu (\bar{x}_2 - a) \left(\bar{x}_{1} + \tilde{x} \bar{x}_2\right)  \nonumber \\
&\qquad -\lambda\mu(\bar{x}_{1b1}-a_0)\left({\bf x}_2 \cdot \rOneBUnit  - \bar{x}_{1b1}\tilde{x}\right) 
-\lambda\mu(\bar{x}_{2b2}-a_0)\left({\bf x}_1 \cdot \rTwoBUnit  - \bar{x}_{2b2}\tilde{x}\right).
\end{align}

%
%
%%%%%%%%%%%%%%%%%%%%%%%%%%%%%%%%%%
\subsection{Proof of {\bf Proposition \ref{prop:existenceProp_TwoAgentsOneBeacon_a0_nonzero}}}
\label{sec:ProofOfProposition42}
%
%
% \begin{proof}
%
It directly follows from \eqref{eqn:2AgentsOneBeaconDynamicsAtEquilibrium} that $\dot{\tilde{x}}=0$ if $\bar{x}_1 = 0$, and in that situation we can express the closed loop dynamics on the nullclines $\dot{\rho} = \dot{\rho}_{1b1} = \dot{\rho}_{2b2} = \dot{\tilde{x}}=0$ as
\begin{align}
\text{\normalsize{ $
\dot{\bar{x}}_1
$}}
&
\text{\normalsize{ $
= (1-\lambda)\mu a + \lambda\mu a_0\left(\frac{\rho_{1b1}^2 +\rho^2 - \rho_{2b2}^2 }{2\rho \rho_{1b1}} \right) + \frac{1}{\rho} \Bigl(1 - \tilde{x}\Bigr)
$}}
,  \nonumber \\
\text{\normalsize{ $
\dot{\bar{x}}_2
$}}
&
\text{\normalsize{ $
= (1-\lambda)\mu a + \lambda\mu a_0\left(\frac{\rho_{2b2}^2 +\rho^2 - \rho_{1b1}^2 }{2\rho \rho_{2b2}} \right) + \frac{1}{\rho} \Bigl(1 - \tilde{x}\Bigr)
$}}, 
\label{eqn:2AgentDynamicsAtEquilibriumCommonCBParameter_x1_zero}
\\
\text{\normalsize{ $
\dot{\bar{x}}_{1b1} 
$}}
&
\text{\normalsize{ $
= 
(1-\lambda) \mu a\left(\frac{\rho_{1b1}^2 + \rho^2 - \rho_{2b2}^2}{2\rho \rho_{1b1}} \right)  
+ \lambda\mu a_0 + \frac{1}{\rho_{1b1}}
$}}
, \nonumber \\
\text{\normalsize{ $
\dot{\bar{x}}_{2b2} 
$}}
&
\text{\normalsize{ $
= 
(1-\lambda) \mu a  \left(\frac{\rho_{2b2}^2 + \rho^2 - \rho_{1b1}^2}{2\rho \rho_{2b2}} \right)  
+ \lambda\mu a_0 + \frac{1}{\rho_{2b2}}
$}}. \nonumber
\end{align}

We note that taking the difference of $\dot{\bar{x}}_1-\dot{\bar{x}}_2$ yields
\begin{align}
\label{eqn:x1DotDiffSupp}
\dot{\bar{x}}_1-\dot{\bar{x}}_2 
% &= 
% \lambda\mu a_0 \left( \frac{\rho_{1b1}^2 +\rho^2 - \rho_{2b2}^2 }{2\rho \rho_{1b1}} - \frac{\rho_{2b2}^2 +\rho^2 - \rho_{1b1}^2 }{2\rho \rho_{2b2}} \right) 
% \nonumber \\
%
%
% &= \frac{\lambda\mu a_0}{2\rho\rho_{1b1}\rho_{2b2}}\Bigl(-\rho^2 (\rho_{1b1}-\rho_{2b2}) + \rho_{2b2}\left(\rho_{1b1}^2 - \rho_{2b2}^2 \right) + \rho_{1b1}\left(\rho_{1b1}^2 - \rho_{2b2}^2 \right) \Bigr) \nonumber \\
% &= \frac{\lambda\mu a_0(\rho_{1b1}-\rho_{2b2}) }{2\rho\rho_{1b1}\rho_{2b2}}\Bigl(-\rho^2 + \rho_{2b2}\left(\rho_{1b1} + \rho_{2b2} \right) + \rho_{1b1}\left(\rho_{1b1} + \rho_{2b2} \right) \Bigr) \nonumber \\
&= 
\frac{\lambda\mu a_0(\rho_{1b1}-\rho_{2b2}) }{2\rho\rho_{1b1}\rho_{2b2}} \Bigl(\left(\rho_{1b1} + \rho_{2b2} \right)^2 -\rho^2 \Bigr),
\end{align}
and similar calculations lead to 
\begin{align}
\label{eqn:x1bDotDiff}
\dot{\bar{x}}_{1b1}-\dot{\bar{x}}_{2b2} 
&= 
(1-\lambda) \mu a \left(\frac{\rho_{1b1}-\rho_{2b2}}{2\rho\rho_{1b1}\rho_{2b2}}\right) 
 \Bigl( (\rho_{1b1} + \rho_{2b2})^2 -\rho^2 \Bigr) + \frac{\rho_{2b2}-\rho_{1b1}}{\rho_{1b1}\rho_{2b2}}.
\end{align}
Then by setting both \eqref{eqn:x1DotDiffSupp} and \eqref{eqn:x1bDotDiff} equal to zero, i.e. by setting the derivatives of $\bar{x}_2$ and $\bar{x}_{2b2}$ identical to the derivatives of $\bar{x}_1$ and $\bar{x}_{1b1}$ respectively, we can conclude that $\rho_{2b2}$ must be equal to $\rho_{1b1}$ at an equilibrium. Substituting this equivalence into \eqref{eqn:2AgentDynamicsAtEquilibriumCommonCBParameter_x1_zero}, we can further conclude that the following conditions must hold true at an equilibrium
\begin{align}
(1-\lambda)\mu a + \lambda\mu a_0 \left( \frac{\rho}{2\rho_{1b1}} \right) + \frac{1 - \tilde{x}}{\rho} 
&= 0,
\label{eqn:equilibriumDynamicsTwoAgentA0NonZeroTerm_1}
\\
(1-\lambda) \left(\frac{\mu a \rho}{2 \rho_{1b1}} \right)  + \lambda\mu a_0 + \frac{1}{\rho_{1b1}}
&= 0.
\label{eqn:equilibriumDynamicsTwoAgentA0NonZeroTerm_2}
\end{align}

\textbf{If the two agents and the beacon are collinear}, then the constraint $\rho_{1b1} = \rho_{2b2}$ implies that $\rho = 2\rho_{1b1}$. Substituting this equivalence into 
\eqref{eqn:equilibriumDynamicsTwoAgentA0NonZeroTerm_1} and \eqref{eqn:equilibriumDynamicsTwoAgentA0NonZeroTerm_2} yields
\begin{align}
(1-\lambda)\mu a + \lambda\mu a_0  + \left(\frac{1 - \tilde{x}}{\rho} \right)
&= 0,
\label{eqn:equilibriumDynamicsTwoAgentA0NonZeroTerm_1_subs}
\\
(1-\lambda)\mu a + \lambda\mu a_0  + \frac{1}{\rho_{1b1}}
&= 0,
\label{eqn:equilibriumDynamicsTwoAgentA0NonZeroTerm_2_subs}
\end{align}
from which it follows that 
\begin{align}
\label{eqn:rhoB11TwoAgentsOneBeaconProof}
\rho_{1b1} &= \frac{\lambda}{-\mu \bigl((1-\lambda)a + \lambda a_0\bigr)},
\end{align}
which is a valid solution if and only if $(1-\lambda)a + \lambda a_0 < 0$. Substitution of $\rho = 2\rho_{1b1}$ (with $\rho_{1b1}$ given by \eqref{eqn:rhoB11TwoAgentsOneBeaconProof}) into \eqref{eqn:equilibriumDynamicsTwoAgentA0NonZeroTerm_1_subs} yields $\tilde{x}=-1$.

\textbf{If the agents and beacon are not collinear}, then the analysis in \cite{KSG_BD_ACC_2018} demonstrates that the equilibrium constraint $\bar{x}_1=\bar{x}_2=0$ implies that $\tilde{x}$ must be either $1$ or $-1$ at such an equilibrium.

\underline{If $\tilde{x}=1$}, then \eqref{eqn:equilibriumDynamicsTwoAgentA0NonZeroTerm_1} allows us to express $\rho$ as
\begin{equation}
\rho 
= 
-2 \left(\frac{1-\lambda}{\lambda} \right)\left(\frac{a}{a_0} \right)\rho_{1b1}.
\label{Existence_4.2_Cond1}
\end{equation}
As both $\rho$ and $\rho_{1b1}$ must be positive, \eqref{Existence_4.2_Cond1} is meaningful if and only if $a/a_0 < 0$. Also, since $\rho_{1b1} = \rho_{2b2}$, substituting \eqref{Existence_4.2_Cond1} into constraint \eqref{eqn:rUnitb1Unit} yields 
\begin{align}
\left(\frac{1-\lambda}{\lambda}\right) \left(\frac{-a}{a_0}\right) < 1,
\label{eqn:prop42partBconstraint}
\end{align}
with the strict inequality resulting from the fact that we have assumed that the agents and beacon are not collinear. The combination of \eqref{eqn:prop42partBconstraint} with $a/a_0 < 0$ yields two possibilities:
\begin{itemize}
\item Case 1: $a_0>0$, $a<0$, $(1-\lambda) a + \lambda a_0 > 0$;
\item Case 2: $a_0<0$, $a>0$, $(1-\lambda) a + \lambda a_0 < 0$.
\end{itemize}
Also, substitution of \eqref{Existence_4.2_Cond1} into \eqref{eqn:equilibriumDynamicsTwoAgentA0NonZeroTerm_2} leads to
% \begin{equation}
% -\mu\left(\frac{(1-\lambda)^2}{\lambda} \right) \left(\frac{a^2}{a_0} \right) + \lambda\mu a_0 + \frac{1}{\rho_{1b1}}
% =
% 0,
% \end{equation}
% which in turn yields
\begin{align}
\rho_{1b1} 
= \frac{\lambda a_0}{\mu\Bigl((1-\lambda)^2 a^2 - \lambda^2 a_0^2\Bigr)},
\end{align}
which yields a valid solution  (i.e. $\rho_{1b1}>0$) if and only if 
% \begin{equation}
% a_0\Bigl((1-\lambda)^2 a^2 - \lambda^2 a_0^2\Bigr) > 0, 
% \end{equation}
% i.e. 
\begin{equation}
\label{eqn:Prop2Constraint}
a_0\Bigl((1-\lambda) a - \lambda a_0\Bigr)\Bigl((1-\lambda) a + \lambda a_0\Bigr) > 0. 
\end{equation}
It is straightforward to verify that Case 2 (but not Case 1) satisfies \eqref{eqn:Prop2Constraint}, leading to part (b) of the proposition.

On the other hand, \underline{if $\tilde{x}=-1$}, then \eqref{eqn:equilibriumDynamicsTwoAgentA0NonZeroTerm_1}-\eqref{eqn:equilibriumDynamicsTwoAgentA0NonZeroTerm_2} simplifies to
\begin{align}
\frac{1}{2\rho\rho_{1b1}}\Bigl[2(1-\lambda)\mu a \rho \rho_{1b1}  + \lambda \mu a_0 \rho^2 + 4\rho_{1b1}\Bigr]
&=
0,
\label{Existence_4.2_Cond2}
\\
\frac{1}{2\rho_{1b1}}\Bigl[(1-\lambda)\mu a \rho  + 2 + 2\lambda\mu a_0 \rho_{1b1} \Bigr]
&=
0.
\label{Existence_4.2_Cond3}
\end{align}
However, it follows from \eqref{Existence_4.2_Cond2}-\eqref{Existence_4.2_Cond3} that at an equilibrium we must have $\rho = 2\rho_{1b1}$, i.e. this corresponds to the collinear configuration addressed earlier in part (a) of the proposition. (Note that the condition $\tilde{x}=\pm 1$ is a necessary condition of the agents and beacon being in a collinear configuration, but it is not sufficient.) This completes the proof.
%
%
% \end{proof}

%
%%%%%%%%%%%%%%%%%%%%%%%%%%%%%%%%%%
\subsection{Proof of {\bf Lemma \ref{lem:tildeXPlusMinusOne}}}
First, suppose that $\hat{\bf b}$ is not parallel to $\mathbf{r}$, i.e. $\rUnit \cdot \bUnit \neq \pm 1$. Then $\bar{x}_{1} = 0 = \bar{x}_{2}$ implies that ${\bf x}_1 \cdot \rUnit = 0$ and ${\bf x}_2 \cdot \rUnit = 0$, and $\hat{x}_{1} = 0 = \hat{x}_{2}$ implies that ${\bf x}_1 \cdot \bUnit = 0$ and ${\bf x}_2 \cdot \bUnit = 0$. Thus we have 
\begin{equation}
    {\bf x}_1 \cdot \left(\rUnit \times \bUnit\right) = \pm 1, \qquad 
    {\bf x}_2 \cdot \left(\rUnit \times \bUnit\right) = \pm 1,
\end{equation}
and since $\hat{\bf b}$ is not parallel to $\mathbf{r}$, it must be that $\tilde{x} = {\bf x}_1 \cdot {\bf x}_2 = \pm 1$.

Now suppose $\hat{\bf b}$ is parallel to $\mathbf{r}$. In a manner similar to the first part of the proof, we can use the assumptions $\bar{x}_{1b} = 0 = \bar{x}_{2b}$ and $\hat{x}_{1} = 0 = \hat{x}_{2}$ to arrive at
\begin{equation}
\label{eqn:crossProductTerms}
    {\bf x}_1 \cdot \left(\rOneBUnit \times \bUnit\right) = \pm 1, \qquad 
    {\bf x}_2 \cdot \left(\rTwoBUnit \times \bUnit\right) = \pm 1.
\end{equation}
From \eqref{eqn:closureConstraint} we have $\mathbf{r}_{2b2} = \mathbf{r}_{1b1} - \mathbf{r} - \hat{\bf b}$, and substitution into the second equation in \eqref{eqn:crossProductTerms} yields
\begin{align}
    \left(\frac{1}{\rho_{2b2}} \right) \left[{\bf x}_2 \cdot \left(\left(\mathbf{r}_{1b1} - \mathbf{r} - \hat{\bf b}\right) \times \bUnit\right)\right]
    =
    \left(\frac{\rho_{1b1}}{\rho_{2b2}} \right) \left[{\bf x}_2 \cdot \left(\rOneBUnit \times \bUnit\right)\right]
    =
    \pm 1,
\end{align}
where we have used the assumption that $\hat{\bf b}$ is parallel to $\mathbf{r}$ and therefore the cross product of the two vectors is zero. Observing that ${\bf x}_2$, $\rOneBUnit$, and $\bUnit$ are all unit vectors, we note that the term inside the brackets must have an absolute value of $1$. Therefore it must hold that $\rho_{1b1} = \rho_{2b2}$ and 
\begin{equation}
\label{eqn:crossProductTerm2}
    {\bf x}_2 \cdot \left(\rOneBUnit \times \bUnit\right) = \pm 1.
\end{equation}
Then since since $\hat{\bf b}$ is not parallel to $\mathbf{r}_{1b1}$, \eqref{eqn:crossProductTerm2} together with the first equation in \eqref{eqn:crossProductTerms} implies that $\tilde{x} = \pm 1$.

\end{document}